\documentclass[12pt,draftcls,onecolumn]{IEEEtran}
\usepackage{graphicx}

\usepackage[draft]{prelim2e}

\usepackage[sort,compress]{cite}
\usepackage[cmex10]{mathtools}
\usepackage{dsfont,amssymb}
\usepackage[standard,amsmath,thmmarks] {ntheorem}

\newtheorem{assumption}     {Assumption}
\let\ALPHABET \mathcal

\let\VEC      \mathbf

\newcommand\IND{\mathds{1}}

\newcommand\DEFINED{\coloneqq}

\usepackage[T1]{fontenc}

\usepackage{xspace}
\newcommand{\ind}[1]{\IND_{\{#1\}}}

\newcommand{\beqq}[1]{\begin{eqnarray*} #1 \end{eqnarray*}}

\newcommand{\pr}[1]{\mathds{P}\left\{#1\right\}}

\newcommand{\payoff}[2]{\renewcommand{\arraystretch}{1}\begin{array}{c|c|c|}
\multicolumn{1}{r}{}
 &  \multicolumn{1}{c}{\text{0}}
 & \multicolumn{1}{c}{\text{1}} \\
\cline{2-3}
\text{0} & #1\\
\cline{2-3}
\text{1} & #2\\
\cline{2-3}
\end{array}}
\begin{document}

\title {Nash Equilibria for Stochastic Games with Asymmetric Information-Part 1: Finite Games}

\author {Ashutosh Nayyar, Abhishek Gupta, C\'edric Langbort and Tamer Ba\c{s}ar
\thanks{The authors are with Coordinated Science Laboratory at the University of Illinois at Urbana-Champaign. Email:
        {\tt\small \{anayyar,gupta54,langbort,basar1\}@illinois.edu}}%
}

\maketitle
\begin{abstract}
A model of  stochastic games where multiple controllers jointly control the evolution of the state of a dynamic system but have access to different information about the state and action processes is considered. The asymmetry of information among the controllers makes it difficult to compute or characterize Nash equilibria. Using common information among the controllers, the game with asymmetric information is shown to be equivalent to another game with symmetric information. Further, under certain conditions, a Markov state is identified for the equivalent symmetric information game and its Markov perfect equilibria are characterized. This characterization provides a backward induction algorithm to find Nash equilibria of the original game with asymmetric information in pure or behavioral strategies. Each step of this algorithm involves finding Bayesian Nash equilibria of a one-stage Bayesian game. The class of Nash equilibria of the original game that can be characterized in this backward manner are named \emph{common information based Markov perfect equilibria}.
\end{abstract}
\section{Introduction}
\label{sec:introduction}
 Stochastic games model situations where multiple players jointly control the evolution of the state of a stochastic dynamic system with each player trying to minimize its own costs. Stochastic games where all players have perfect state observation are well-studied \cite{Shapley:1953,Sobel1971, Tirole,Basarbook,Filarbook}. In such games, the symmetry of information among players implies that they  all share the same uncertainty about the future states and future payoffs. However, a number of games arising in communication systems, queuing systems, economics, and in models of adversarial interactions in  control and communication systems involve players with \emph{different}  information about the state and action processes. Due to the asymmetry of information, the players have different beliefs about the current state and different uncertainties about future states and payoffs. As a result, the analytical tools for finding Nash equilibria for  stochastic games with perfect state observation cannot be directly employed for games with asymmetric information.
\par
In the absence of a general framework for  stochastic games with asymmetric information, several special models have been studied in the literature. In particular, zero-sum differential games with linear dynamics and quadratic payoffs where the two players  have different observation processes were studied in \cite{Behn68}, \cite{Rhodes69}, \cite{Willman69}. A zero sum differential game where one player's observation at any time includes the other player's observation was considered in \cite{Hojota}. A zero-sum differential game where one player has a noisy observation of the state while the other controller has no observation of the state was considered in  \cite{MintzBasar}. Discrete-time non-zero sum LQG games with one step delayed sharing of observations were studied in \cite{Basaronestep}, \cite{Basarmulti}. A one-step delay observation and action sharing game was considered in \cite{Altman:2009}. A two-player finite game in which the players do not obtain each other's observations and  control actions
was considered in \cite{hespanha2001} and  a necessary and sufficient condition for Nash equilibrium in terms of two coupled dynamic programs was presented.
\par
 Obtaining equilibrium solutions for  stochastic  games  when players make independent noisy observations of the state and do not  share all of their information (or even when they have access to the same noisy observation as in \cite{Bjota}) has remained a challenging problem for general classes of games. Identifying classes of game structures which would lead to tractable solutions or feasible solution methods is therefore an important goal in that area. In this paper, we identify one such class of nonzero-sum    stochastic games, and obtain characterization of a class of Nash equilibrium strategies. 
%\par
%For finite state dynamic games with asymmetric information, a new notion of equilibrium called ``Applied Markov Equilibrium'' was defined in \cite{Fershtman}. This equilibrium concept is based on the concept of self-confirming equilibrium  which, in general, differs from Nash equilibrium \cite{FudenbergLevine}.
\par 
In  stochastic games with perfect state observation, a subclass of Nash equilibria - namely the Markov perfect equilibria- can be obtained by backward induction. The advantage of this technique is that instead of searching for equilibrium in the (large) space of strategies, we only need to find Nash equilibrium in a succession of static games of complete information. 
\par
Can a backward inductive decomposition be extended to games of asymmetric information? The general answer to this question is negative. However, we show that there is a class of asymmetric information games that are amenable to such a decomposition. The basic conceptual observation underlying our results is the following: the essential impediment to applying backward induction in asymmetric information games is the fact that a player's posterior beliefs about the system state and about other players' information may depend on the strategies used by the players in the past. If the nature of system dynamics and the information structure of the game ensures that the players' posterior beliefs are strategy independent, then a backward induction argument is feasible. We formalize this conceptual argument in this paper. 
\par
 We first use the common information among the controllers to show that the game with asymmetric information is equivalent to another game with symmetric information. Further, under the assumption of strategy independence of posterior beliefs, we identify a Markov state for the equivalent symmetric information game and characterize its Markov perfect equilibria using backward induction arguments. This characterization provides a backward induction algorithm to find Nash equilibria of the original game with asymmetric information. Each step of this algorithm involves finding Bayesian Nash equilibria of a one-stage Bayesian game. The class of Nash equilibria of the original game that can be characterized in this backward manner are named \emph{common information based Markov perfect equilibria}. For notational convenience, we consider games with only two controllers. Our results extend to games with $n>2$ controllers in a straightforward manner.
 \par
 Our work is conceptually similar to the work in \cite{Cole2001}. The authors in \cite{Cole2001} considered a model of  finite stochastic  game with discounted infinite-horizon cost function  where each player has a privately
observed state. Under the assumption that player $i$'s belief about other players'
states depends only the current state of player $i$ and does not depend
on player $i$'s strategy, \cite{Cole2001} presented a recursive algorithm to compute  Nash
equilibrium. Both our model and our main assumptions differ from those in \cite{Cole2001}.

%the standard technique for finding equilibrium is a backwards inductive approach that ``decomposes'' the multi-stage game into a several one stage games using backward inductive arguments. At each stage, the payoffs/costs of the future are accounted for by means of  value functions which describe the expected future costs as a function of the next state. The equilibrium strategies found in this manner are Markov strategies (that is, they only use current state as argument of the control law) and form a sub-game perfect equilibrium (that is, they are an equilibrium of each sub-game of the original game).  Such decompositions, in general, cannot be obtained if the controllers have asymmetric information. We want to present conditions under which an asymmetric information game admits a suitable backward inductive decomposition and consider examples (like one step delay LQG, CDC 2012 paper) where explicit solutions can be found. 

\subsection{Notation}
Random variables are denoted by upper case letters; their realizations
by the corresponding lower case letters. Random vectors are denoted by upper case bold letters and their 
realizations by lower case bold letters. Unless otherwise stated, the state, action and observations are assumed to be vector valued. Subscripts are used as time index. $\VEC X_{a:b}$ is a short hand for
the vector $(\VEC X_a, \VEC X_{a+1}, \dots, \VEC X_b)$, if $a>b$, then $\VEC X_{a:b}$ is empty.   $\mathds{P}(\cdot)$ is the
probability of an event, $\mathds{E}(\cdot)$ is the expectation of a random
variable. For a collection of functions $\boldsymbol{g}$, 
$\mathds{P}^{\boldsymbol{g}}(\cdot)$ and $\mathds{E}^{\boldsymbol{g}}(\cdot)$ 
denote that the probability/expectation depends on the choice
of functions in $\boldsymbol{g}$. Similarly, for a probability distribution $\pi$,  $\mathds{E}^{\pi}(\cdot)$ 
denotes that the expectation is with respect to the distribution $\pi$. The notation $\mathds{1}_{\{a=b\}}$ denotes $1$ if the equality in the subscript is true and $0$ otherwise. For a finite set $\mathcal{A}$,  $\Delta(\mathcal{A})$ is the set of all probability mass functions over $\mathcal{A}$. 
For two random variables (or random vectors) $X$ and $Y$, $\mathds{P}(X=x|Y)$ denotes the conditional
probability of the event $\{X=x\}$ given $Y$. This is a random variable whose realization depends on the realization of $Y$.

When dealing with collections of random variables, we will at times treat the collection as a random vector of appropriate dimension. At other times, it will be convenient to think of different collections of random variables as sets on which one can define the usual set operations. For example consider random vectors $\VEC A = (A_1,A_2,A_3)$ and $\tilde{\VEC A} = (A_1,A_2)$. Then, treating $\VEC A$ and $\tilde{\VEC A}$ as sets would allow us to write $\VEC A \setminus \tilde{\VEC A} = \{A_3\}$.
\subsection{Organization}
The rest of this paper is organized as follows. We present our model of a  stochastic game with asymmetric information in Section~\ref{sec:model}. We present several special cases of our model in Section~\ref{sec:examples}. We prove our main results in Section~\ref{sec:virtual}. We extend our arguments to consider behavioral strategies in Section~\ref{sec:behave}. We examine the importance of our assumptions in Section~\ref{sec:comments}. Finally, we conclude in Section~\ref{sec:conclusion}.
%\section{Problem Formulation} \label{sec:PF}

% For the ease of exposition, we describe the model in discrete time and assume
% that the system runs for a finite horizon and all variables are finite. We refer
% to this model as the \emph{basic model}. We show
% how to extend the model to infinite horizon in Section~\ref{sec:infinite}.  With
% appropriate technical assumptions, the model and the results also apply to
% continuous time systems and continuous valued random variables.

\section{The Basic Game \textbf{G1}}
\label{sec:model}
\subsection{The Primitive Random Variables and the Dynamic System} We consider a collection of finitely-valued, mutually independent random vectors $(\VEC X_1, \VEC W^0_1,$ $\VEC W^0_2,\ldots, \VEC W^0_{T-1},$ $\VEC W^1_1, \VEC W^1_2,\ldots, \VEC W^1_T,$ $\VEC W^2_1, \VEC W^2_2,\ldots, \VEC W^2_T)$ with known  probability mass functions. These random variables are called the primitive random variables. 

%(Assume all random variables are finitely valued for now.)
%\subsection{The Dynamic System} 
We consider a discrete-time dynamic system with $2$ controllers. For any time $t$, $t=1,2,\ldots,T$, $\VEC X_t \in \ALPHABET X_t$ denotes the state
of the system at time $t$, $\VEC U^i_t \in \ALPHABET U^i_t$ denotes the control action
of controller~$i$, $i=1,2$ at time $t$. The state of the system evolves according to
\begin{equation}
  \label{eq:state}
  \VEC X_{t+1} = f_t(\VEC X_t, \VEC U^1_t, \VEC U^2_t, \VEC W^0_t).
\end{equation}
There are two observation processes: $\VEC Y^1_t \in \ALPHABET Y^1_t, \VEC Y^2_t \in \ALPHABET Y^2_t$, where 
    \begin{equation}  \label{eq:observation}
      \VEC Y^i_t = h^i_t(\VEC X_t, \VEC W^i_t), \mbox{~~}i=1,2.
    \end{equation}
    
\subsection{The Data Available to Controllers}  At any time  $t$, the vector $\VEC I^i_t$ denotes the  total data available to controller $i$ at time $t$. The vector $\VEC I^i_t$ is a subset of the collection of \emph{potential observables} of the system at time $t$, that is, $\VEC I^i_t \subset \{\VEC Y^1_{1:t}, \VEC Y^2_{1:t}, \VEC U^1_{1:t-1}, \VEC U^2_{1:t-1}\}$. We  divide the total data into two components: \emph{private information $\VEC P^i_t$ and common information $\VEC C_t$.} Thus, $\VEC I^i_t = (\VEC P^i_t, \VEC C_t)$. As their names suggest, the common information is available to both controllers whereas private information is available only to one controller. Clearly, this separation of information into private and common part can always be  done. In some cases, common or private information may even be empty.  For example, if $\VEC I^1_t = \VEC I^2_t = \{\VEC Y^1_{1:t}, \VEC Y^2_{1:t}, \VEC U^1_{1:t-1}, \VEC U^2_{1:t-1}\}$, that is if both controllers have access to all observations and actions, then $\VEC C_t = \VEC I^1_t= \VEC I^2_t$ and $\VEC P^1_t=\VEC P^2_t= \emptyset$. On the other hand, if $\VEC I^i_t=\VEC Y^i_{1:t}$, for $i=1,2$, then $\VEC C_t = \emptyset$ and $\VEC P^i_t = \VEC I^i_t$. Games where are all information is common to both controllers are referred to as symmetric information games. 

 We denote the set of possible realizations of $\VEC P^i_t$ as $\mathcal{P}^i_t$ and the set of possible realizations of $\VEC C_t$ as $\mathcal{C}_t$. Controller~$i$ chooses action $\VEC U^i_t$ as a function of the total data $(\VEC P^i_t, \VEC C_t)$ available to it.  Specifically, for each controller~$i$,
\begin{equation}
  \label{eq:control}
  \VEC U^i_t = g^i_t(\VEC P^i_t, \VEC C_t),
\end{equation}
where $g^i_t$, referred to as the control law at time $t$, can be any function of private and common information.
The collection $\VEC g^i = (g^i_1, \dots, g^i_T)$ is called the
\emph{control strategy} of controller~$i$ and the pair of control strategies for the two controllers $(\VEC g^1, \VEC g^2)$ is called a \emph{strategy profile}. 
For a given strategy profile, the overall cost of controller $i$ is given as 
\begin{equation}\label{eq:costofi}
J^i(\VEC g^1,  \VEC g^2) := \mathds{E}\Big[\sum_{t=1}^T c^i(\VEC X_t,\VEC U^1_t, \VEC U^2_t)\Big],
\end{equation}
where the expectation on the right hand side of \eqref{eq:costofi} is with respect to the probability measure  on the state and action processes induced by the choice of strategies $\VEC g^1, \VEC g^2$ on the left hand side of \eqref{eq:costofi}.
A strategy profile $(\VEC g^1,\VEC g^2)$ is called a Nash equilibrium if no controller can lower its total expected cost by unilaterally changing its strategy, that is,
\begin{equation}\label{eq:nashcond1} 
J^1(\VEC g^1,  \VEC g^2) \leq  J^1( \mathbf{\tilde{g}^1}, \VEC g^2), \quad \text{and}\quad J^2(\VEC g^1,  \VEC g^2) \leq  J^2(\VEC g^1, \mathbf{\tilde{g}^2}),
\end{equation}
%\begin{equation}\label{eq:nashcond1.1} 
%J^2(\VEC g^1,  \VEC g^2) \leq  J^2(\VEC g^1, \mathbf{\tilde{g}^2}),
%\end{equation}
for all strategies $\mathbf{\tilde{g}^1}, \mathbf{\tilde{g}^2}$. We refer to the above game as game \textbf{G1}.
\begin{Remark}
The system dynamics and the observation model (that is, the functions $f_t, h^1_t,h^2_t$ in \eqref{eq:state} and \eqref{eq:observation}), the statistics of the primitive random variables, the information structure of the game and the cost functions are assumed to be common knowledge among the controllers.
\end{Remark}
\subsection{Evolution of Common and Private Information}\label{sec:assumption}

%We make the following assumptions on the information structure of the game:
\begin{assumption}\label{assm:infoevolution}
We assume that the common and private information evolve over time as follows:
\begin{enumerate}
\item The common information $\VEC C_t$ is increasing with time, that is, $\VEC C_t \subset \VEC C_{t+1}$ for all $t$. 
Let $\VEC Z_{t+1} = \VEC C_{t+1}\setminus \VEC C_t$ be the increment in common information from time $t$ to $t+1$. Thus, $\VEC C_{t+1} = \{\VEC C_t, \VEC Z_{t+1}\}$. Further, 
\begin{equation}
\VEC Z_{t+1} = \zeta_{t+1}(\VEC P^1_t, \VEC P^2_t, \VEC U^1_t, \VEC U^2_t, \VEC Y^1_{t+1}, \VEC Y^2_{t+1}),\label{eq:commoninfo}
\end{equation}
where $\zeta_{t+1}$ is a fixed transformation.

\item The private information evolves according to the equation
\begin{equation}
\VEC P^i_{t+1} = \xi^i_{t+1}(\VEC P^i_t, \VEC U^i_t,  \VEC Y^i_{t+1}) \label{eq:privateinfo}
\end{equation}
where $\xi^i_{t+1}, i=1,2,$ are fixed transformations.
%We denote the set of possible realizations of $\VEC P^i_t$ as $\mathcal{P}^i_t$. %Later in the paper, we will present several instances where our assumptions on the evolution of common and private information are true.
\end{enumerate}
\end{assumption}

Equation \eqref{eq:commoninfo} states that the increment in common information is a function of the ``new'' variables generated between $t$ and $t+1$, that is, the actions taken at $t$ and the observations made at $t+1$, and the ``old'' variables that were part of private information at time $t$. Equation \eqref{eq:privateinfo} implies that the evolution of private information at the two controllers is influenced by different observations and actions.

\subsection{Common Information Based Conditional Beliefs}
A key concept in our analysis is the belief about the state and the private informations conditioned on the common information of both controllers. %Since this belief is  based on the common information, it is common knowledge among the controllers.
 Formally,
at any time $t$, given the control laws from time $1$ to $t-1$, we define the common information based conditional belief as follows:
\begin{equation}
\Pi_t(\VEC x_t, \VEC p^1_t, \VEC p^2_t) := \mathds{P}^{g^1_{1:t-1},g^2_{1:t-1}}(\VEC X_t =\VEC x_t, \VEC P^1_t = \VEC p^1_t, \VEC P^2_t = \VEC p^2_t|\VEC C_t) ~~~~ \text{for all~} \VEC x_t, \VEC p^1_t, \VEC p^2_t, \label{eq:beliefeq1}
\end{equation}
where we use the superscript $g^1_{1:t-1},g^2_{1:t-1}$ in the RHS of \eqref{eq:beliefeq1} to emphasize that the conditional belief depends on the past control laws. Note that $\Pi_t(\cdot,\cdot,\cdot)$ is a $|\mathcal{X}_t \times \mathcal{P}^1_t \times \mathcal{P}^2_t|$-dimensional random vector whose realization depends on the realization of $\VEC C_t$. A realization of $\Pi_t$ is denoted by $\pi_t$.

Given control laws $g^1_t,g^2_t$, we define the following partial functions:
\[\Gamma^1_t = g^1_t(\cdot, \VEC C_t) \mbox{~~~} \Gamma^2_t = g^2_t(\cdot, \VEC C_t) \]
These partial functions are functions from the private information of a controller to its control action. These are \emph{random functions} whose realizations depend on the realization of the random vector $\VEC C_t$.
The following lemma describes the evolution of the common information based conditional belief using these partial functions.

%The following lemma describes the evolution of the beliefs $Pi_t$.
\begin{lemma}\label{lemma:evolution}
Consider any choice of control laws $g^1_{1:t},g^2_{1:t}$. Let $\pi_t$ be the realization of the common information based conditional belief at time $t$, let $\VEC c_t$ be the realization of the common information at time $t$, let $\gamma^i_t = g^i_t(\cdot,\VEC c_t)$, $i=1,2$, be the corresponding realizations of the partial functions at time $t$, and $\VEC z_{t+1}$ be the realization of the increment in common information (see  Assumption \ref{assm:infoevolution}). Then, the realization of the conditional belief at time $t+1$ is given as
\begin{equation}
\pi_{t+1} = F_t(\pi_t,\gamma^1_t,\gamma^2_t,\VEC z_{t+1}), \label{eq:evolution}
\end{equation}
where $F_t$ is a fixed transformation that does not depend on the control strategies.
\end{lemma}
\begin{proof} See Appendix \ref{sec:update_func}.
\end{proof}

Lemma \ref{lemma:evolution} states that the evolution of the conditional belief $\Pi_t$ is governed by the partial functions of control laws at time $t$. This lemma relies on Assumption \ref{assm:infoevolution}  made earlier about the evolution of common and private information. We now introduce the following critical assumption that eliminates the dependence of $\Pi_t$ on the control laws.

%\subsection{Strategy Independence of Beliefs}
%So far we have not imposed any assumptions on the system dynamics  or specified the exact nature of common and private information. 

\begin{assumption}[Strategy Independence of Beliefs] \label{assm:separation}
 Consider any time $t$, any choice of control laws $g^1_{1:t-1},g^2_{1:t-1}$, and any realization of common information $\VEC c_t$ that has a non-zero probability  under $g^1_{1:t-1},g^2_{1:t-1}$. Consider any other choice of control laws $\tilde{g}^1_{1:t-1},\tilde{g}^2_{1:t-1}$ which also gives a non-zero probability  to $\VEC c_t$. Then, we assume that
 \[\mathds{P}^{g^1_{1:t-1},g^2_{1:t-1}}(\VEC X_t =\VEC x_t, \VEC P^1_t = \VEC p^1_t, \VEC P^2_t = \VEC p^2_t|\VEC c_t) = \mathds{P}^{\tilde{g}^1_{1:t-1},\tilde{g}^2_{1:t-1}}(\VEC X_t =\VEC x_t, \VEC P^1_t = \VEC p^1_t, \VEC P^2_t = \VEC p^2_t|\VEC c_t),\]
 for all $\VEC x_t, \VEC p^1_t, \VEC p^2_t$.
 \par
 Equivalently,  the evolution of the common information based conditional belief described in Lemma \ref{lemma:evolution} depends only on the increment in common information, that is,  \eqref{eq:evolution} can be written as
\begin{equation}
\pi_{t+1} = F_t(\pi_t,\VEC z_{t+1}), \label{eq:evolution1}
\end{equation}
where $F_t$ is a fixed transformation that does not depend on the control strategies. 
  \end{assumption}
\par
\begin{Remark}
 Assumption \ref{assm:separation} is somewhat related to the notion of one-way separation in stochastic control, that is, the estimation (of the state in standard stochastic control and of the state and private information in Assumption \ref{assm:separation}) is independent of the control strategy.
\end{Remark}
%\par
%We will describe in Section \ref{sec:examples} instances of game \textbf{G1} where the nature of the dynamic system and the private and common information implies that Assumptions \ref{assm:infoevolution} and \ref{assm:separation} hold. 

%In the next section, we will analyze game \textbf{G1} under the assumption A1. In Sections 3 ..., we will consider specific models where Assumption 1 holds. We will use the analysis of the next section to find Nash equilibria for these specific models.
%%%%%%%%%%%%%%%%%%%%%%%%%%%%%%%%%%%%%%%%%%%%%%%%%%%%%%%%%%%%%%%%%%%%%%%%%%%%%%%%%%%%%%%%%%%%%%%%%%%%%%%%%%%%%%%%%%%
\section{Games satisfying Assumptions \ref{assm:infoevolution} and \ref{assm:separation}} \label{sec:examples}
Before proceeding with further analysis, we first describe some instances of \textbf{G1} where the nature of the dynamic system and the private and common information implies that Assumptions \ref{assm:infoevolution} and \ref{assm:separation} hold.

\subsection{One-Step Delayed Information Sharing Pattern}\label{sec:onestepmodel}
Consider the instance of \textbf{G1} where the common information at any time $t$ is given as $\VEC C_t = \{\VEC Y^1_{1:t-1}, \VEC Y^2_{1:t-1},\VEC U^1_{1:t-1},\VEC U^2_{1:t-1}\}$ and the private information is given as $\VEC P^i_t = \VEC Y^i_t$. Thus, $\VEC Z_{t+1} := \VEC C_{t+1} \setminus \VEC C_t = \{\VEC Y^1_t, \VEC Y^2_t, \VEC U^1_t, \VEC U^2_t\}$. This information structure can be interpreted as the case where all observations and actions are shared among controllers with one step delay. %It is straightforward to verify that this structure of common and private information satisfies Assumption~\ref{assm:infoevolution}. The following lemma shows that the common information based belief satisfies Assumption~\ref{assm:separation} in this case.
\begin{lemma}\label{lemma:onestep}
The game with one-step delayed sharing information pattern described above satisfies Assumptions~\ref{assm:infoevolution} and \ref{assm:separation}.
%In the one-step delayed sharing information pattern described above, given a realization of the common information $\VEC c_t$, the common information based belief $\pi_t$ does not depend on the choice of control laws. In particular, the common information based belief satisfies an update equation $\pi_{t+1} = F_t(\pi_t,\VEC z_{t+1})$ where $F_t$ does not depend on the control laws.
\end{lemma}
\begin{proof} See Appendix \ref{sec:onestep}\footnote{Appendices F-J are included in the Supplementary Material section at the end of the paper.}. 
\end{proof}
%Because of Lemma \ref{lemma:onestep}, we can use Theorem \ref{thm:backward_ind} to find a common information based Markov perfect equilibrium for the game with one step delayed information sharing pattern. %We will consider a linear quadratic Gaussian game with this information structure in Section~\ref{sec:lqgonestep}.
\par
A special case of the above information structure is the situation where the state $\VEC X_t = (X^1_t,X^2_t)$ and controller $i$'s observation $Y^i_t = X^i_t$. A game with this information structure was considered in \cite{Altman:2009}. 
It is interesting to note that Assumption~\ref{assm:separation} is not true if information is shared with delays larger than one time step \cite{NMT:2011}. %Our analysis, therefore, does not extend to  delays larger than one time step. 

\subsection{Information Sharing with One-Directional-One-Step Delay}
Similar to the one-step delay case, we consider the situation where all observations of controller 1 are available to controller 2 with no delay while the observations of controller 2 are available to controller 1 with one-step delay. All past control actions are available to both controllers. That is, in this case, $\VEC C_t = \{\VEC Y^1_{1:t}, \VEC Y^2_{1:t-1},\VEC U^1_{1:t-1},\VEC U^2_{1:t-1}\}$, $\VEC Z_{t+1} = \{\VEC Y^1_{t+1},\VEC Y^2_t,\VEC U^1_t,\VEC U^2_t\}$, controller 1 has no private information  and the private information of controller 2  is  $\VEC P^2_t = \VEC Y^2_t$. %Assumption~\ref{assm:infoevolution} is clearly satisfied. The following lemma establishes the validity of Assumption~\ref{assm:separation} as well.
\begin{lemma}\label{lemma:asymmetric_delay}
The game with one-directional-one-step delayed sharing information pattern described above satisfies Assumptions~\ref{assm:infoevolution} and \ref{assm:separation}.
%In the one directional one-step delayed sharing information pattern described above, given a realization of the common information $\VEC c_t$, the common information based belief $\pi_t$ does not depend on the choice of control laws. In particular, the common information based belief satisfies an update equation $\pi_{t+1} = F_t(\pi_t,\VEC z_{t+1})$ where $F_t$ does not depend on the control laws.
\end{lemma}
\begin{proof} See Appendix \ref{sec:asymmetric_delay}. 
\end{proof}
%\subsection{An Information Structure with one-way sharing}
%next, we consider the case where controller 1 observes the state $\VEC X_t$ perfectly, whereas controller 2 has noisy observations of the state. We assume that the observations of controller 2 are available to controller 1 but not vice versa. As before, all control actions are available to both controllers. Thus, $\VEC C_t = \{\VEC Y^2_{1:t}, \VEC U^1_{1:t-1},\VEC U^2_{1:t-1}\}$, $\VEC P^1_t = \VEC X_t$ and controller 2 has no private information. (Note that the private information of controller 1 is just the current state and not the entire history of states.)
%%%%%%%%%%%%%%%%%%%%%%%%%%%%%%%%%%%%%%%%%%%%%\input new_example
\subsection{State Controlled by One Controller with Asymmetric Delay Sharing}
\emph{Case A:} Consider the special case of \textbf{G1} where the state dynamics are controlled only by controller $1$, that is,
\[ \VEC X_{t+1} = f_t(\VEC X_t, \VEC U^1_t, \VEC W^0_t).\]
Assume that the information structure is given as:
\[\VEC C_t = \{\VEC Y^1_{1:t},\VEC Y^2_{1:t-d}, \VEC U^1_{1:t-1}\}, \qquad \VEC P^1_t = \emptyset, \qquad \VEC P^2_t = \VEC Y^2_{t-d+1:t}.\]
That is, controller $1$'s observations  are available to controller $2$ instantly while controller $2$'s observations are available to controller $1$ with a delay of $d \geq 1$ time steps. 

\emph{Case B:} Similar to the above case, consider the situation where the state dynamics are still controlled only by controller $1$ but the information structure is:
\[\VEC C_t = \{\VEC Y^1_{1:t-1},\VEC Y^2_{1:t-d}, \VEC U^1_{1:t-1}\}, \qquad \VEC P^1_t = \VEC Y^1_t, \qquad \VEC P^2_t = \VEC Y^2_{t-d+1:t}.\]
%Assumption~\ref{assm:infoevolution} is clearly satisfied in both the cases above. The following lemma establishes the validity of Assumption~\ref{assm:separation} as well.
\begin{lemma}\label{lemma:new_example}
The games described in Cases A and B satisfy Assumptions~\ref{assm:infoevolution} and \ref{assm:separation}.
%In the games described in Cases A and B above, given a realization of the common information $\VEC c_t$, the common information based belief $\pi_t$ does not depend on the choice of control laws. In particular, the common information based belief satisfies an update equation $\pi_{t+1} = F_t(\pi_t,\VEC z_{t+1})$ where $F_t$ does not depend on the control laws.
\end{lemma}
\begin{proof} See Appendix \ref{sec:new_example}. 
\end{proof}

%%%%%%%%%%%%%%%%%%%%%%%%%%%%%%%%%%%%%%%%%%%%%%%%%%%%%%%%%%%%%%%%%
\subsection{An Information Structure with Global and Local States}
\emph{Noiseless Observations:} We now consider the information structure described in \cite{NayyarBasarcdc}. In this example, the state $\VEC X_{t}$ has three components: a \emph{global state} $X^0_t$ and a \emph{local state} $X^i_t$ for each controller. The state evolution is given by the following equation:
\begin{equation}\label{eq:dynamics1}
\VEC X_{t+1} = f_t(X^0_t, \VEC U^1_t, \VEC U^2_t, \VEC W^0_t)
\end{equation}
Note that the dynamics depend on the current global state $X^0_t$ but not on the current local states.
Each controller has access to the global state process $X^0_{1:t}$ and its current local state $X^i_{t}$. In addition, each controller knows the past actions of all controllers.  Thus, the common and private information in this case are: 
\[ \VEC C_t = \{X^0_{1:t}, \VEC U^1_{1:t-1}, \VEC U^2_{1:t-1}\}, ~~~~~~ \VEC P^i_t =\{X^i_t\}\]
It is straightforward to verify that Assumption \ref{assm:infoevolution} holds for this case.

For a realization $\{x^0_{1:t}, \VEC u^1_{1:t-1}, \VEC u^2_{1:t-1}\}$ of the common information, the common information based belief in this case is given as
\begin{align} &\pi_t(x^0,x^1,x^2) = \mathds{P}^{g^1_{1:t-1},g^2_{1:t-1}}(X^0_t=x^0,x^1_t=x^1,X^2_t=x^2|x^0_{1:t}, \VEC u^1_{1:t-1}, \VEC u^2_{1:t-1}) \notag \\
&= \mathds{1}_{\{x^0=x^0_t\}}\mathds{P}(X^1_t=x^1,X^2_t=x^2|x^0_{t:t-1},\VEC u^1_{t-1},\VEC u^2_{t-1})
\end{align}
It is easy to verify that the above belief depends only on the statistics of $\VEC W^0_{t-1}$ and is therefore independent of control laws. Thus, Assumption~\ref{assm:separation} also holds for this case.
\par
\emph{Noisy Observations:} We can also consider a modification of the above scenario where both controllers have a common, noisy observation $Y^0_t = h_t(X^0_t,W^1_t)$ of the global state. That is,
\[ \VEC C_t = \{Y^0_{1:t}, \VEC U^1_{1:t-1}, \VEC U^2_{1:t-1}\}, \qquad \VEC P^i_t =\{X^i_t\}, \qquad \VEC Z_{t+1} = \{Y^0_{t+1}, \VEC U^1_t, \VEC U^2_t\}.\]
%For a realization $\{y^0_{1:t}, \VEC u^1_{1:t-1}, \VEC u^2_{1:t-1}\}$ of the common information, the common information based belief in this case is given as
%\begin{align} &\pi_t(x^0,x^1,x^2) = \mathds{P}^{g^1_{1:t-1},g^2_{1:t-1}}(X^0_t=x^0,x^1_t=x^1,X^2_t=x^2|y^0_{1:t}, \VEC u^1_{1:t-1}, \VEC u^2_{1:t-1})
%\end{align}
%The following lemma establishes the validity of Assumption~\ref{assm:separation}.
\begin{lemma}\label{lemma:noisy_global}
The game with the information pattern described above satisfies Assumptions~\ref{assm:infoevolution} and \ref{assm:separation}.
\end{lemma}
\begin{proof} See Appendix \ref{sec:noisy_global}. 
\end{proof}
% \notag \\
%&= \sum_{x' \in \mathcal{X}_{t-1}}\mathds{P}(X^1_t=x^1,X^2_t=x^2|X^0_t=x^0,X^0_{t-1}=x'\VEC u^1_{t-1},\VEC u^2_{1:t-1})\mathds{P}^{g^1_{1:t-1},g^2_{1:t-1}}(X^0_t = x^0, X^0_{t-1} =x'|y^0_{1:t}, \VEC u^1_{1:t-1}, \VEC u^2_{1:t-1})\label{eq:july19.1}
%\end{align}
%The first term in the right hand side of \eqref{eq:july19.1} depends only on the statistics of $\VEC W^0_{t-1}$ and the following lemma shows that the second term is control strategy independent.
%The following lemma establishes the validity of Assumption~\ref{assm:separation} as well.
%\begin{lemma}
%In the  information pattern described above, given a realization of the common information $\VEC c_t$, the common information based belief $\pi_t$ does not depend on the choice of control laws. In particular, the common information based belief satisfies an update equation $\pi_{t+1} = F_t(\pi_t,\VEC z_{t+1})$ where $F_t$ does not depend on the control laws.
%\end{lemma}
%

\subsection{Uncontrolled State Process} \label{sec:uncontrolled}
Consider a state process whose evolution does not depend on the control actions, that is, the system state evolves as
\begin{equation}
\VEC X_{t+1} = f_t(\VEC X_t, \VEC W^0_t)
\end{equation}
Further, the common and private information evolve as follows:
\begin{enumerate}
\item  $\VEC C_{t+1} = \{\VEC C_t, \VEC Z_{t+1}\}$ and 
\begin{equation}
\VEC Z_{t+1} = \zeta_{t+1}(\VEC P^1_t, \VEC P^2_t, \VEC Y^1_{t+1}, \VEC Y^2_{t+1}),\label{eq:uncontrolled_1}
\end{equation}
where $\zeta_{t+1}$ is a fixed transformation.

\item The private information evolves according to the equation
\begin{equation}
\VEC P^i_{t+1} = \xi^i_{t+1}(\VEC P^i_t,   \VEC Y^i_{t+1}) \label{eq:uncontrolled_2}
\end{equation}
where $\xi^i_{t+1},i=1,2,$ are fixed transformations.
\end{enumerate}
Note that while control actions do not affect the state evolution, they still affect the costs.
%Following the arguments in the proof of Lemma~\ref{lemma:evolution}, we can prove the following lemma which ensures that Assumption~\ref{assm:separation} is true for this case.
\begin{lemma}\label{lemma:uncontrolled}
The game \textbf{G1} with an uncontrolled state process described above satisfies Assumptions~\ref{assm:infoevolution} and \ref{assm:separation}.
\end{lemma}
\begin{proof}
See Appendix \ref{sec:uncontrolled_app}.
\end{proof}
As an example of this case, consider the information structure where the two controllers share their observations about an uncontrolled state process with a delay of $d \geq 1$ time steps. In this case, the common information is $\VEC C_t = \{\VEC Y^1_{1:t-d},\VEC Y^2_{1:t-d}\}$ and the private information is $\VEC P^i_t = \VEC Y^i_{t-d+1:t}$. % Because Assumptions~\ref{assm:infoevolution} and \ref{assm:separation} are true for this example, we can use Theorem \ref{thm:backward_ind} to find a common information based Markov perfect equilibrium of this game.
%\subsection{A Learning Game with Asymmetric Information}
%Consider the case where t

\subsection{Symmetric Information Game}
Consider the case when all observations and actions are available to both controllers, that is, $\VEC I^1_t = \VEC I^2_t = \VEC C_t = \{\VEC Y^1_{1:t}, \VEC Y^2_{1:t}, \VEC U^1_{1:t-1},\VEC U^2_{1:t-1}\}$ and there is no private information. %As mentioned earlier, the prescriptions in this case are simply control actions (instead of being functions from private information to control actions).
 The common information based belief in this case is $\pi_{t}(\VEC x_{t}) = \mathds{P}^{g^1_{1:t-1},g^2_{1:t-1}}(\VEC X_{t} =\VEC x_{t}|\VEC y^1_{1:t},\VEC y^2_{1:t},\VEC u^1_{1:t-1},\VEC u^2_{1:t-1})$.\\ $\pi_t$ is the same as the information state in centralized stochastic control, which is known to be control strategy independent and which satisfies an update equation of the form required in Assumption~\ref{assm:separation} \cite{KumarVaraiya}. A related case with perfect state observations is the situation where $\VEC I^1_t = \VEC I^2_t = \VEC X_{1:t}$. %Thus, the result of Theorem \ref{thm:backward_ind} holds for this game as well. However, since there is no private information, the one stage Bayesian games in each step of the backward inductive process used in Theorem \ref{thm:backward_ind} now reduce to one-stage games of complete information.
 \subsection{Symmetrically Observed Controlled State and Asymmetrically Observed Uncontrolled State}
A combination of the previous two scenarios is the situation when the state $\VEC X_t$ consists of two independent components: a controlled component $X^a_t$   and an uncontrolled component $X^b_t$. Both components are observed through noisy channels. The observations about the controlled state as well as the past actions are common to both controllers whereas the information about the uncontrolled state satisfies the model of Section \ref{sec:uncontrolled}. The common information based conditional belief can then be factored into two independent components each of which satisfies an update equation of the form required by Assumption~\ref{assm:separation}. 
%%%%%%%%%%%%%%%%%%%%%%%%%%%%%%%%%%%%%%%%%%%%%%%%%%%%%%%%%%%%%%%%%%%%%%%%%%%%%%%%%%%%%%%%%%%%%%%%%%%%%%%%%%%%%%%%%%
\section{Main Results}\label{sec:virtual}
%Game \textbf{G1} is a game of \emph{asymmetric information} since at any time $t$ the data available to the two controller is different. In order to characterize equilibrium of this game, we want to formulate a game with \emph{symmetric information} whose equilibrium can be used to find an equilibrium of game \textbf{G1}. We will do this by replacing each controller with a virtual player and ensuring that both virtual players have the same information.
%\par
Our goal in this section is to show  that under  Assumptions \ref{assm:infoevolution} and \ref{assm:separation}, a class of equilibria of the game \textbf{G1}   can be characterized in a backward inductive manner that resembles the backward inductive characterization of Markov perfect equilibria of {symmetric information games with perfect state observation}. However, in order to do so, we have to view our asymmetric information game as a symmetric information game by introducing ``virtual players'' that make decisions based on the common information. This section describes this change of perspective and how it can be used to characterize a class of Nash equilibria.
\par
We reconsider the model of game \textbf{G1}. We assume that controller $i$ is replaced by a virtual player $i$ (VP $i$).  The system operates as follows: At time $t$, the data available to each virtual player is the common information $\VEC C_t$. The virtual player $i$ selects a \emph{function} $\Gamma^i_t$ from $\mathcal{P}^i_t$ to $\mathcal{U}^i_t$ according to a decision rule $\chi^i_t$,
\[ \Gamma^i_t = \chi^i_t(\VEC C_t) \]
Note that under a given decision rule $\chi^i_t$, $\Gamma^i_t$ is a random function since $\VEC C_t$ is a  random vector.
 We will use $\gamma^i_t$ to denote a realization of $\Gamma^i_t$. We will refer to $\Gamma^i_t$ as the \emph{prescription} selected by virtual player $i$ at time $t$.  Once the virtual player has chosen $\Gamma^i_t$, a control action $\VEC U^i_t = \Gamma^i_t(\VEC P^i_t)$ is applied to the system.
$\VEC \chi^i := (\chi^i_1,\chi^i_2,\ldots, \chi^i_T)$ is called the strategy of the virtual player $i$.  The total cost of the virtual player $i$ is given as
\begin{equation}
  \label{eq:newcost1}
  \mathcal{J}^i( \VEC \chi^1, \VEC \chi^2 ) \DEFINED \mathds{E}\Big[ 
  \sum_{t=1}^T c^i(\VEC X_t, \VEC U^1_t, \VEC U^2_t) \Big]
\end{equation}
where the expectation on the right hand side of \eqref{eq:newcost1} is with respect to the probability measure  on the state and action processes induced by the choice of strategies $\VEC \chi^1, \VEC \chi^2$ on the left hand side of \eqref{eq:newcost1}. We refer to the game among the virtual players as game \textbf{G2}.
\begin{Remark} \label{remark:empty}
In case there is no private information, the function $\Gamma^i_t$ from $\mathcal{P}^i_t$ to $\mathcal{U}^i_t$ is interpreted as simply a value in the set $\mathcal{U}^i_t$.
\end{Remark}
\subsection{Equivalence with Game \textbf{G1}}
\begin{theorem}\label{thm:equiv}
Let $(\VEC g^1, \VEC g^2)$ be a Nash equilibrium of game \textbf{G1}. Define $\VEC \chi^i$ for $i=1,2$, $t=1,2,\ldots,T$ as
\begin{equation} \chi^i_t(\VEC c_t) := g^i_t(\cdot, \VEC c_t), \label{eq:g12g2}
\end{equation}
for each possible realization  $\VEC c_t$ of common information at time $t$.
Then $(\VEC \chi^1, \VEC \chi^2)$ is a Nash equilibrium of game \textbf{G2}. Conversely, if $(\VEC \chi^1, \VEC \chi^2)$ is a Nash equilibrium of game \textbf{G2}, then define $\VEC g^i$ for $i=1,2,$ $t=1,2,\ldots,T$ as
\begin{equation} g^i_t(\cdot,\VEC c_t) := \chi^i_t(\VEC c_t), 
\end{equation}
for each possible realization  $\VEC c_t$ of common information at time $t$.
Then $(\VEC g^1,\VEC g^2)$ is a Nash equilibrium of game \textbf{G1}. 
\end{theorem}
\begin{proof}
It is clear that using \eqref{eq:g12g2}, any controller strategy profile in game \textbf{G1} can be transformed to a corresponding virtual player strategy profile in game \textbf{G2}  without altering the behavior of the dynamic system and in particular the values of the expected costs. If a virtual player can reduce its costs by unilaterally deviating from $\chi^i$, then such a deviation must also exist for the corresponding controller in \textbf{G1}. Therefore, equilibrium of controllers' strategies implies equilibrium of corresponding virtual players' strategies. The converse can be shown using similar arguments. %Similarly, any virtual player strategy profile in \textbf{G2} can be transformed to a controller strategy profile in \textbf{G1} without changing the values of expected costs. Theorem \ref{thm:equiv} is simply a consequence of this equivalence between the two games.
\end{proof}
\par
The game between the virtual players is a symmetric information game since they both make their decisions based only on the common information $\VEC C_t$. In the next section, we identify a Markov state for this symmetric information game and characterize Markov perfect equilibria for this game.  

\subsection{Markov Perfect Equilibrium of \textbf{G2}}\label{sec:VPMPE}
%Some discussion about information state??
%\par
We  want to establish that the common information based conditional beliefs $\Pi_t$ (defined in \eqref{eq:beliefeq1}) can serve as a Markov state for the game \textbf{G2}. Firstly, note that because of Assumption \ref{assm:separation}, $\Pi_t$ depends only on the common information $\VEC C_t$ and since both the virtual players know the common information, the belief $\Pi_t$ is common knowledge among them. The following lemma shows that $\Pi_t$ evolves as a controlled Markov process.
\begin{lemma}\label{lemma:markovprop}
From the virtual players' perspective, the process $\Pi_t, t=1,2,\ldots,T$ is a controlled Markov process with the virtual players' prescriptions $\gamma^1_t,\gamma^2_t, t=1,2,\ldots,T$ as the controlling actions, that is,
% conditioned on the current state of $\pi_t$ and the current prescriptions $\gamma^1_t,\gamma^2_t$, the next state $\Pi_{t+1}$ is independent of the past beliefs and past prescriptions $\pi_{1:t-1},\gamma^1_{1:t-1},\gamma^2_{1:t-1}$. In particular,
\begin{align}
\mathds{P}(\Pi_{t+1}|\VEC c_t, \pi_{1:t},\gamma^1_{1:t},\gamma^2_{1:t}) = \mathds{P}(\Pi_{t+1}| \pi_{1:t},\gamma^1_{1:t},\gamma^2_{1:t}) =\mathds{P}(\Pi_{t+1}|\pi_{t},\gamma^1_{t},\gamma^2_{t})
\end{align}
\end{lemma}
\begin{proof}
See Appendix \ref{sec:markovprop}.
\end{proof}
 Following the development in \cite{MaskinTirole}, we next show that if one virtual player is using a strategy that is measurable with respect to $\Pi_t$, then the other virtual player can select an optimal response strategy measurable with respect to $\Pi_t$ as well. %The following lemma shows that $\Pi_t$ is indeed an information state of game \textbf{G2}.

\begin{lemma}\label{lemma:infostatelemma}
If virtual player $i$ is using a decision strategy that selects  prescriptions only as a function of the belief $\Pi_t$, that is,
\[ \Gamma^i_t = \psi^i_t(\Pi_t), \]
$t=1,\ldots,T,$ then virtual player $j$ can also choose its prescriptions only as a function of the belief $\Pi_t$ without any loss of performance. 
\end{lemma}
\begin{proof}
See Appendix \ref{sec:infostatelemma}
\end{proof}
%\subsubsection{Markov perfect equilibria of Game \textbf{G2}}
Lemmas \ref{lemma:markovprop} and \ref{lemma:infostatelemma} establish $\Pi_t$ as the Markov state for the game \textbf{G2}. We  now define a Markov perfect equilibrium for game \textbf{G2}.

%For any time $t$, and any realization $\VEC c_t$ of the common information at time $t$ we denote by $SG_t(\VEC c_t)$ the sub-game of the game \textbf{G2} starting from time $t$ with $\VEC c_t$ as the realization of common information and with the costs associated with the 
% we define the sub-game $SG_t(\VEC c_t)$ as the game between virtual players from time $t$ to $T$  

\begin{definition}
A strategy profile $(\psi^1,\psi^2)$ is said to be a Markov perfect equilibrium of game \textbf{G2} if (i) at each time $t$, the strategies select prescriptions only as a function of the common information based belief $\Pi_t$ and (ii) the strategies form a Nash equilibrium for every sub-game of \textbf{G2} \cite{Tirole}.
\end{definition}

Given a Markov perfect equilibrium of \textbf{G2}, we can construct a corresponding Nash equilibrium of game \textbf{G1} using Theorem~\ref{thm:equiv}. We refer to the class of Nash equilibria of \textbf{G1} that can be constructed from the Markov perfect equilibria of \textbf{G2} as the \emph{common information based Markov perfect equilibria of} \textbf{G1}. 
\begin{definition}
 A strategy profile $(\VEC g^1, \VEC g^2)$ of the form $\VEC U^i_t = g^i_t(\VEC P^i_t, \Pi_t), i=1,2,$  is called a \emph{common information based Markov perfect   equilibrium} for game \textbf{G1} if the corresponding strategies of game \textbf{G2} defined as
\[  \psi^i_t(\pi_t):= g^i_t(\cdot, \pi_t), \]
form a Markov perfect equilibrium of \textbf{G2}.
\end{definition}
The following theorem provides a necessary and sufficient condition for a strategy profile to be a Markov perfect equilibrium of \textbf{G2}.
\begin{theorem}\label{thm:equlb_condition}
Consider a strategy pair $(\psi^1,\psi^2)$ such that at each time $t$, the strategies select prescriptions based only on the realization  of the common information based belief $\pi_t$, that is,
\[ \gamma^i_t = \psi^i_t(\pi_t), ~~~~~~~~ i=1,2\]
A necessary and sufficient condition for $(\psi^1,\psi^2)$ to be a Markov perfect equilibrium of \textbf{G2} is that they satisfy the following 
conditions:
\begin{enumerate}
\item For each possible realization $\pi$ of $\Pi_T$, define  the value function for virtual player $1$:
\begin{align}
  &V^1_T(\pi) := \min_{\tilde{\gamma}^1} \mathds{E}[c^1(\VEC X_t, {\Gamma}^1_T(\VEC P^1_T), \Gamma^{2}_T(\VEC P^{2}_T))|\Pi_T=\pi,\Gamma^1_T = \tilde{\gamma}^1, \Gamma^2_T = \psi^{2}_T(\pi)] \label{eq:dpeq1} 
\end{align}
 Then, $\psi^1_T(\pi)$ must be a minimizing $\tilde{\gamma}^1$ in the definition of $V^1_T(\pi)$. 
Similarly, define  the value function for virtual player $2$:
\begin{align}
  &V^2_T(\pi) := \min_{\tilde{\gamma}^2} \mathds{E}[c^2(\VEC X_t, {\Gamma}^1_T(\VEC P^1_T), {\Gamma}^{2}_T(\VEC P^{2}_T))|\Pi_T=\pi,\Gamma^1_t = \psi^{1}_T(\pi), \Gamma^2_T = \tilde{\gamma}^2] \label{eq:dpeq2}
\end{align}
 Then, $\psi^2_T(\pi)$ must be a minimizing $\tilde{\gamma}^2$ in the definition of $V^2_T(\pi)$.
~~\medskip \\
\item For $t=T-1,\ldots,1$ and for each possible realization $\pi$ of $\Pi_t$, define recursively the value functions for virtual player $1$:
\begin{align}
  &V^1_t(\pi) := \min_{\tilde{\gamma}^1}\mathds{E}[c^1(\VEC X_t, \Gamma^1_t(\VEC P^1_t), \Gamma^{2}_t(\VEC P^{2}_t)) +  V^1_{t+1}(\Pi_{t+1})|\Pi_t=\pi, \Gamma^1_t = \tilde{\gamma}^1, \Gamma^2_t = \psi^{2}_t(\pi)]\label{eq:dpeq3}
  \end{align}
  where  $\Pi_{t+1} = F_t(\Pi_t,\VEC Z_{t+1})$. %and $\VEC Z_{t+1}$ is the increment in common information generated according to \eqref{eq:commoninfo}, \eqref{eq:observation} and \eqref{eq:state} when control actions $\VEC U^i_t =\tilde{\gamma}^1(\VEC P^1_t)$ and $\VEC U^2_t = \gamma^{2}_t(\VEC P^{2}_t)$ are used.
   Then, $\psi^1_t(\pi)$ must be a minimizing $\tilde{\gamma}^1$ in the definition of $V^1_t(\pi)$. Similarly, define recursively the value functions for virtual player $2$:
\begin{align}
  &V^2_t(\pi) := \min_{\tilde{\gamma}^2}\mathds{E}[c^2(\VEC X_t, \Gamma^1_t(\VEC P^1_t), {\Gamma}^{2}(\VEC P^{2}_t)) +  V^2_{t+1}(\Pi_{t+1})|\Pi_t=\pi,  \Gamma^1_t = \psi^{1}_t(\pi), \Gamma^2_t = \tilde{\gamma}^2]\label{eq:dpeq4}
  \end{align}
  where $\Pi_{t+1} = F_t(\Pi_t,\VEC Z_{t+1})$. Then, $\psi^2_t(\pi)$ must be a minimizing $\tilde{\gamma}^2$ in the definition of $V^2_t(\pi)$.
\end{enumerate}
\end{theorem}
\begin{proof}
See Appendix \ref{sec:equlb_condition}
\end{proof}
%We want to consider a stronger equilibrium concept for the game \textbf{G2} that is similar to Markov perfect equilibrium of symmetric games. In particular, we would like to find decision strategies $\psi^1, \psi^2$ such that:
%\begin{enumerate}
%\item The strategies select prescriptions based on the common belief $\Pi_t$. (Markov property)
%\item At each time $t$ and each possible realization $\pi$ of the belief $\Pi_t$, the strategies form a Nash equilibrium of the sub-game starting from time $t$ with the belief $\pi$.
%\end{enumerate}

%(The above properties are essentially trying to imitate the concept of Markov perfect equilibrium of symmetric information games, that is, the strategies should form a sub-game perfect equilibrium and only use the relevant history. This needs further formalization.)
%\par
Theorem~\ref{thm:equlb_condition} suggest that one could follow a backward inductive procedure to find equilibrium strategies for the virtual players. Before describing this backward procedure in detail, we make a simple but useful observation. In \eqref{eq:dpeq1}-\eqref{eq:dpeq4}, since the $\tilde{\gamma}^i$ enters the expectation only as $\tilde{\gamma}^i(\VEC P^i)$, it suggests that we may be able to carry out the minimization over $\tilde{\gamma}^i$ by separately minimizing over $\tilde{\gamma}^i(\VEC p^i)$ for all possible $\VEC p^i$. This observation leads us to the backward induction procedure described in the next section.
\begin{Remark}
Note that if Assumption~\ref{assm:separation} were not true, then according to Lemma \ref{lemma:evolution}, $\Pi_{t+1} = F_t(\Pi_t,\Gamma^1_t,\Gamma^2_t,\VEC Z_{t+1})$. In this case, the entire prescription $\tilde{\gamma}^i$ will affect the second term in the expectation in \eqref{eq:dpeq3}-\eqref{eq:dpeq4}, and we could not hope to carry out the minimization over $\tilde{\gamma}^i$ by separately minimizing over $\tilde{\gamma}^i(\VEC p^i)$ for all possible $\VEC p^i$.
\end{Remark}
\subsection{Backward Induction Algorithm for Finding Equilibrium}\label{sec:backward_ind}
We can now describe a backward inductive procedure to find a Markov perfect equilibrium of game \textbf{G2} using a sequence of one-stage Bayesian games. We proceed as follows: \\
\underline{\textbf{Algorithm 1:}}
\begin{enumerate}
\item At the terminal time $T$, for each realization $\pi$ of the common information based belief at time $T$, we define a one-stage Bayesian game  $SG_T(\pi)$ where 
\begin{enumerate}
\item The probability distribution on $(\VEC X_T, \VEC P^1_T, \VEC P^2_T)$ is $\pi$.
\item  Agent\footnote{Agent $i$ can be thought to be the same as controller $i$. We use a different name here in order to maintain the distinction between games \textbf{G1} and $SG_T(\pi)$.} $i$  observes $\VEC P^i_T$ and chooses action $\VEC U^i_T$, $i=1,2$. 
\item  Agent $i$'s cost is $c^i(\VEC X_T,\VEC U^1_T,\VEC U^2_T)$, $i=1,2$. 
\end{enumerate}
A Bayesian Nash equilibrium of this game is a pair of strategies $\gamma^i,i=1,2,$ for the agents which map their observation $\VEC P^i_T$ to their  action $\VEC U^i_T$ such that for any realization $\VEC p^i$, $\gamma^i(\VEC p^i)$ is a solution of the minimization problem 
\[ \min_{\VEC u^i} \mathds{E}^{\pi}[c^i(\VEC X_T,\VEC u^i, \gamma^{j}(\VEC P^{j}_T))|\VEC P^i_T = \VEC p^i ], \]
where $j \neq i$ and the superscript $\pi$ denotes that the expectation is with respect to the distribution $\pi$. 
(See \cite{Osborne,Myerson_gametheory} for a definition of Bayesian Nash equilibrium.)
 If a Bayesian Nash equilibrium $\gamma^{1*},\gamma^{2*}$ of $SG_T(\pi)$ exists, denote the corresponding expected equilibrium costs as $V^i_T(\pi), i=1,2$ and define $\psi^i_T(\pi) := \gamma^{i*}$, $i=1,2$.

\item At time $t < T$, for each realization $\pi$ of the common information based belief at time $t$, we define the one-stage Bayesian   game $SG_t(\pi)$ where
\begin{enumerate}
\item The probability distribution on $(\VEC X_t, \VEC P^1_t, \VEC P^2_t)$ is $\pi$.
\item  Agent $i$ observes $\VEC P^i_t$ and chooses action $\VEC U^i_t$, $i=1,2$. 
\item  Agent $i$'s cost is $c^i(\VEC X_t,\VEC U^1_t,\VEC U^2_t) + V^i_{t+1}(F_t(\pi, \VEC Z_{t+1}))$, $i=1,2$. 
\end{enumerate}
Recall that the belief for the next time step is $\Pi_{t+1} = F_t(\pi, \VEC Z_{t+1})$ and $\VEC Z_{t+1}$ is given by \eqref{eq:commoninfo}.
A Bayesian Nash equilibrium of this game is a pair of strategies $\gamma^i,i=1,2,$ for the agents  which map their observation $\VEC P^i_t$ to their  action $\VEC U^i_t$ such that for any realization $\VEC p^i$, $\gamma^i(\VEC p^i)$ is a solution of the minimization problem 
\[ \min_{\VEC u^i} \mathds{E}^{\pi}[c^i(\VEC X_t,\VEC u^i, \gamma^{j}(\VEC P^{j}_t)) +V^i_{t+1}(F_t(\pi, \VEC Z_{t+1})) |\VEC P^i_t = \VEC p^i ], \]
where $j \neq i, i,j=1,2,$ and $\VEC Z_{t+1}$ is the increment in common information generated according to \eqref{eq:commoninfo}, \eqref{eq:observation} and \eqref{eq:state} when control actions $\VEC U^i_t =\VEC u^i$ and $\VEC U^j_t = \gamma^{j}(\VEC P^{j}_t)$ are used. The expectation is with respect to the distribution $\pi$.
If a Bayesian Nash equilibrium $\gamma^{1*},\gamma^{2*}$ of $SG_t(\pi)$ exists, denote the corresponding expected equilibrium costs as $V^i_t(\pi), i=1,2$ and define $\psi^i_t(\pi) := \gamma^{i*}$, $i=1,2$.
%Consider any Bayesian Nash equilibrium $\gamma^{1*},\gamma^{2*}$ of $SG_t(\pi)$ and denote the equilibrium costs as $V^i_t(\pi)$. Define $\psi^i_t(\pi) := \gamma^{i*}$, $i=1,2$. 
\end{enumerate}
\par
 \begin{theorem}\label{thm:backward_ind}
  The strategies $\psi^1,\psi^2$ defined by the backward induction procedure described in Algorithm 1 form a Markov perfect equilibrium of game \textbf{G2}. Consequently, strategies $\VEC g^1, \VEC g^2$ defined as
\[ g^i_t(\cdot, \pi_t) := \psi^i_t(\pi_t), \]
$i=1,2$, $t=1,2,\ldots,T$ form a common information based Markov perfect equilibrium of game \textbf{G1}.
\end{theorem}
\begin{proof}
To prove the result, we just need to observe that the strategies defined by the backward induction procedure of Algorithm 1  satisfy the  conditions of Theorem \ref{thm:equlb_condition} and hence form a Markov perfect equilibrium of game \textbf{G2}. See Appendix \ref{sec:backward_indproof} for a more detailed proof.%See Appendix \ref{sec:backward_ind}.
\end{proof}

%%%%%%%%%%%%%%%%%%%%%%\input Finite_Example_new%%%%%%%%%%%%%%%%%%%%%%%%%%%%%%%%%
% \section{Finite Example}
\subsection{An Example Illustrating Algorithm 1}\label{sec:Num_example}
We consider an example of game \textbf{G1} where the (scalar) state $X_t$ and the (scalar) control actions $U^1_t,U^2_t$ take value in the set $\{0,1\}$. The state evolves as a controlled Markov chain depending on the two control actions according to the state transition probabilities:
\begin{align}
\pr{ X_{t+1} = 0\big| X_t = 0, U^1_t =  U^2_t}&=\frac{1}{4},\notag\\
\pr{  X_{t+1} = 0\big| X_t = 1, U^1_t = U^2_t}&=\frac{1}{2},\notag \\
\pr{  X_{t+1} = 0\big| X_t =0, U^1_t\neq  U^2_t}&=\pr{  X_{t+1} = 0\big| X_t =1, U^1_t\neq  U^2_t} =\frac{2}{5}.\label{eq:exdynamics}
\end{align}
The initial state is assumed to be equi-probable, i.e., $\pr{X_1=0}=\pr{X_1=1}=1/2$. The first controller observes the state perfectly, while the second controller observes the state through a binary symmetric channel with probability of error $1/3$. Thus,
\[Y^1_t = X_t, \qquad Y^2_t = \left\{\begin{array}{ll} X_t & \text{ with probability } \frac{2}{3},\\ 1-X_t & \text{ with probability }\frac{1}{3}. \end{array}\right.\]
%The behavior of observation channel of the second controller is considered to be i.i.d. across time. %The information available to the two controllers at time $t$ is $\VEC I^1_t = \{Y^1_{1:t},Y^2_{1:t-1},U^{1:2}_{1:t-1}\}$ and $\VEC I^2_t = \{Y^1_{1:t-1}, Y^2_{1:t},U^{1:2}_{1:t-1}\}$.
The controllers share the observations and actions with a delay of one time step. Thus, the common information and private informations at  time step $t$ are given as
\beqq{\VEC C_t = \{ X_{1:t-1}, Y^2_{1:t-1},U^{1}_{1:t-1},U^2_{1:t-1}\} , \quad \VEC P^1_t = \{X_t\},\quad \VEC P^2_t  = \{Y^2_t\}.}
%and the increment in common information $\VEC Z_{t+1} = \VEC C_{t+1}\setminus \VEC C_{t} = \{Y^{1}_{t},U^{1:2}_t\}$. 
%It is easy to verify that the above example satisfies Assumption~\ref{assm:infoevolution}.
In the equivalent game with virtual players,  the decision of the $i^{th}$ virtual player, $\Gamma^i_t$, is a function that maps $\ALPHABET Y^i_t:=\{0,1\}$ to $\ALPHABET U^i_t:=\{\text{0,1}\}$.

%Also, at any time $t$, $\pi_t$ can be one of only two possible distributions which we denote by $\pi^0_t:=\pr{X_t=0,Y^2_t=y^2|C_t}$ and $\pi^1_t:=\pr{X_t=1,Y^2_t=y^2|C_t}$.
The common information based belief for this case is the belief on $(X_t,Y^2_t)$ given the common information $x_{1:t-1}, y^2_{1:t-1},u^{1}_{1:t-1},u^2_{1:t-1}$, that is,
\begin{align}
&\pi_t(x,y^2) = \pr{X_t=x,Y^2_t=y^2\big|x_{1:t-1}, y^2_{1:t-1},u^{1}_{1:t-1},u^2_{1:t-1}} \notag \\
&= \pr{X_t = x\big| x_{t-1}, u^{1}_{t-1},u^2_{t-1}}\left(\frac{2}{3}\ind{y^2=x}+\frac{1}{3}\ind{y^2\neq x} \right). \label{eq:ex_update}
\end{align}
% Since the state at the time step $t$ is dependent only on $\{X_{t-1},U^1_{t-1},U^2_{t-1}\}$, we can write the distributions as
%\beqq{\pi^i_t = \pr{X_t = i,Y^2_t=y^2\big| C_t}=\pr{X_t = i\big| X_{t-1}, U^{1:2}_{t-1}}\left(\frac{2}{3}\ind{y^2=i}+\frac{1}{3}\ind{y^2\neq i} \right), i=0,1.}
%It should be noted that the belief at every time step depends only on two parameters: the state $X_{t-1}$ and the number $\ind{U^1_{t-1}=U^2_{t-1}}$. Thus, there are three different beliefs corresponding to the pairs $(X_{t-1},\ind{U^1_{t-1}=U^2_{t-1}})=(0,1),(1,1)$ and the same belief for pairs $(0,0)$ and $(1,0)$. By a slight abuse of notation, we represent $\pi^i_t$ as a function of the pair $(X_{t-1},\ind{U^1_{t-1}=U^2_{t-1}})$.
The above equation implies that the distribution $\pi_t $ is completely specified by $x_{t-1}, u^{1}_{t-1},u^2_{t-1}$. That is, 
\begin{equation}
\pi_t= F_{t-1}(x_{t-1},u^1_{t-1},u^2_{t-1}). \label{eq:ex_update2}
\end{equation} 
(Note that  $F_{t-1}$ is a vector-valued function whose components are given by \eqref{eq:ex_update} for all $x,y^2 \in \{0,1\}$.)
% Also note that using Bayes' rule, we can write the conditional probability of $X_t$ given the private information $Y^2_t$ of the second player as
%\beqq{\pr{X_t=i|Y^2_t=y^2,\pi_t\left(x_{t-1},\ind{u^1_{t-1}=u^2_{t-1}}\right)} = \frac{\pi^i_t\left(x_{t-1},\ind{u^1_{t-1}=u^2_{t-1}}\right)}{\sum_{i=0}^1 \pi^i_t\left(x_{t-1},\ind{u^1_{t-1}=u^2_{t-1}}\right)}.}
%The belief evolves as
%\beq{\label{eq:ex_update}\pi_{t+1}(x_t,\ind{u^1_t=u^2_t}) &=& \pr{X_{t+1} = x,Y^2_{t+1}=y^2\big| C_{t+1}}\nonumber\\
%& =& \pr{X_{t+1} = x\big| x_{t}, u^{1:2}_{t}}\left(\frac{2}{3}\ind{y^2=x}+\frac{1}{3}\ind{y^2\neq x} \right)\nonumber\\
%&=& F_t(\{x_t,u^{1:2}_{t}\}).}
%Since the state $x_t$ at the time step $t$ is common information at the time step $t+1$, the evolution of belief does not depend on the belief $\pi_t$ at the time step $t$. 
The cost functions $c^i(x,u^1,u^2)$ for various values of state and actions are described by the following matrices
\beqq{\begin{array}{cc} 
\qquad x_t =0 & \qquad x_t =1\\
\payoff{1,0 & 0,1}{0,1 & 0,0}& \payoff{0,0 & 1,1}{0,1 & 1,0}\\
\end{array}\label{eq:costmatrix}}
where the rows in each matrix correspond to controller 1's actions and the columns correspond to controller 2's actions. The first entry in each element of the cost matrix is controller 1's cost and second entry is controller 2's cost.
%===================================================================================================
%===================================================================================================
%===================================================================================================

\underline{\emph{Applying Algorithm 1:}}\\
We now use Algorithm 1 for a two-stage version of the game described above.
\begin{enumerate}
\item At the terminal time step $T=2$, for a realization $\pi$ of the common information based belief at time $2$, we define a one stage game $SG_2(\pi)$ where
\begin{enumerate}
\item  The probability distribution on $(X_2,Y^2_2)$ is $\pi$.
\item  Agent $1$ observes $X_2$ and selects an action $U^1_2$; Agent 2 observes $Y^2_2$ and selects $U^2_2$. 
\item  Agent $i$'s cost is $c^i(X_2,U^1_2,U^2_2)$, given by the matrices defined above. 
\end{enumerate}
A Bayesian Nash equilibrium of this game is a pair of strategies $\gamma^1,\gamma^2$, such that 
 \begin{itemize}
 \item For $x =0,1$, $\gamma^1(x)$ is a solution of $\min_{u^1} \mathbb{E}^{\pi}[c^1(X_2,u^1,\gamma^2(Y^2_2))|X_2=x]$.
 \item For $y=0,1$, $\gamma^2(y)$ is a solution of $ \min_{u^2} \mathbb{E}^{\pi}[c^2(X_2,\gamma^1(X_2),u^2)|Y^2_2=y]$. 
% \item $\gamma^2(1)$ is a solution of the minimization problem 
% \[ \min_{u^2} \mathbb{E}^{\pi^0}[c^2(X_2,\gamma^1,u^2)|Y^2_2=1]. \]
\end{itemize}

It is easy to verify that \[\gamma^1(x) := 1,\quad \gamma^2(y) := 1\text{ for all }x,y\in \{0,1\}\] is a Bayesian Nash equilibrium of $SG_2(\pi)$. %This is because playing $1$ is weakly dominating strategy for the first player irrespective of the state, for which the best response of the second player is to play 1, again irrespective of the state. 
The expected equilibrium cost for agent $i$ is 
\begin{align}
&V^i_2(\pi) = \mathbb{E}^{\pi}[c^i(X_2,1,1)] = \left\{\begin{array}{ll} \pi(X_2=1) & \text{ for } i=1, \\ 0 & \text{ for }i=2 \end{array}\right. \label{eq:ex_stage2costs}
\end{align}
where $\pi(X_2=1)$ is the probability that $X_2 =1$ under the distribution $\pi$.
%where the expectation is with respect to the distribution $\pi$. Thus,
%\begin{align}
%&V^1_2(\pi) = \mathbb{E}^{\pi}[c^1(X_2,1,1)] = 3/4, \notag \\
%&V^2_2(\pi)= \mathbb{E}^{\pi}[c^2(X_2,1,1)]  = 0.
%\end{align}
From the above Bayesian equilibrium strategies, we define the virtual players's decision rules for time $T=2$ as $\psi^i_2(\pi) = \gamma^i$, $i=1,2$.
\item At time $t=1$, since there is no common information, the common information based belief $\pi_1$ is simply the prior belief on $(X_1,Y^2_1)$. Since the initial state is equally likely to be $0$ or $1$, 
\[\pi_1(x,y^2) = \frac{1}{2}\left(\frac{2}{3}\ind{y^2=x}+\frac{1}{3}\ind{y^2\neq x} \right)\]
We define the one-stage Bayesian game $SG_1(\pi_1)$ where 
\begin{enumerate}
\item  The probability distribution on $(X_1,Y^2_1)$ is $\pi_1$.
\item  Agent $1$ observes $X_1$ and selects an action $U^1_1$; Agent 2  observes $Y^2_1$ and selects $U^2_1$. 
\item  Agent $i$'s cost is given by $c^i(X_1, U^1_1, U^2_1) + V^i_{2}(F_1(X_1,U^1_1,U^2_1))$, where $F_1$, defined by \eqref{eq:ex_update2} and \eqref{eq:ex_update}, gives the common information belief at time $2$ as a function of $X_1,U^1_1,U^2_2$, and $V^i_2$, defined in \eqref{eq:ex_stage2costs}, gives the expected equilibrium cost for time $2$  as a function of the common information belief at time $2$. 
\par
 For example, if $U^1_1 \neq U^2_1$, then \eqref{eq:ex_update}, \eqref{eq:ex_update2} and \eqref{eq:ex_stage2costs} imply $V^1_{2}(F_1(X_1,U^1_1,U^2_1)) = 3/5$. Similarly, if $U^1_1 = U^2_1$, then \eqref{eq:ex_update}, \eqref{eq:ex_update2} and \eqref{eq:ex_stage2costs} imply $V^1_{2}(F_1(0,U^1_1,U^2_1)) = 3/4$ and $V^1_{2}(F_1(1,U^1_1,U^2_1)) = 1/2$. Also, \eqref{eq:ex_stage2costs}  implies that $V^2_2$ is identically $0$.
\end{enumerate}
A Bayesian Nash equilibrium of this game is a pair of strategies $\delta^1,\delta^2$ such that 
\begin{itemize}
\item For $x=0,1$, $\delta^1(x)$ is a solution of \[\min_{u^1} \mathbb{E}^{\pi_1}[c^1(X_1,u^1,\delta^2(Y^2_1)) + V^1_2(F_1(X_1,u^1,\delta^2(Y^2_1)))|X_1=x].\]
% where $F_1(\{X_1,U^1_1,U^2_1\})$ is $\pi$ with probability $1/4$ and $\pi^1$ with probability $3/4$ when $u^1 = \delta^2(x^2_1)$ and $F_1(F_1(\{X_1,U^1_1,U^2_1\}))$ is $\pi^0$ with probability $3/4$ and $\pi^1$ with probability $1/4$ when $u^1 \neq \delta^2(x^2_1)$.
\item For $y=0,1$, $\delta^2(y)$ is a solution of 
\[ \min_{u^2} \mathbb{E}^{\pi_1}[c^2(X_1,\delta^1(X_1),u^2) +V^2_2(F_1(X_1,\delta^1(X_1),u^2))|Y^2_1=y]. \]
% where $F_1(F_1(\{X_1,U^1_1,U^2_1\}))$ is $\pi^0$ with probability $1/4$ and $\pi^1$ with probability $3/4$ when $\delta^1 = u^2$ and $F_1(F_1(\{X_1,U^1_1,U^2_1\}))$ is $\pi^0$ with probability $3/4$ and $\pi^1$ with probability $1/4$ when $\delta^1 \neq u^2$.
% \item $\delta^2(1)$ is a solution of the minimization problem 
% \[ \min_{u^2} \mathbb{E}^{\pi^0}[c^2(X_2,\delta^1,u^2) +V^2_2(F_1(F_1(\{X_1,U^1_1,U^2_1\})))|Y^2_2=1], \]
% where $F_1(F_1(\{X_1,U^1_1,U^2_1\}))$ is $\pi^0$ with probability $1/4$ and $\pi^1$ with probability $3/4$ when $\delta^1 = u^2$ and $F_1(F_1(\{X_1,U^1_1,U^2_1\}))$ is $\pi^0$ with probability $3/4$ and $\pi^1$ with probability $1/4$ when $\delta^1 \neq u^2$.
\end{itemize}
%Note that under the prior belief $\pi$ on $X_1,Y^2_1$, given $Y^2_1 =y$ the posteriors of the second agent is $\pr{X_1=0|Y^2_1=0} = 2/3$ and $\pr{X_1=0|Y^2_1=1} = 1/3$.
 It is easy to verify that \[\delta^1(x) = 1-x, \quad \delta^2(y) = 1-y\] is a Bayesian Nash equilibrium of $SG_1(\pi)$. The expected equilibrium costs are 
 \[ V^i_1(\pi_1) = \mathds{E}[c^i(X_1,{\delta}^1(X_1),{\delta}^2(Y^2_1))], \]
 which gives
$V^1_1(\pi_1) = 47/60$ and $V^2_1(\pi_1) =1/3$. 
From the above Bayesian equilibrium strategies, we define the virtual players's decision rules for time $t=1$ as $\psi^i_1(\pi_1) = \delta^i$, $i=1,2$.
 
Since we now know the equilibrium decision rules $\psi^i_t$, $i=1,2,t=1,2$ for the virtual players, we can construct the corresponding control laws for the controllers using Theorem~\ref{thm:backward_ind}. Thus, a common information based Markov perfect equilibrium for the game in this example is given by the strategies:
\beqq{g^1_1(x_1,\pi_1) = \left\{\begin{array}{ll}\text{1} & \text{if } x_1 = 0,\\\text{0} & \text{if } x_1 = 1.  \end{array}\right. \qquad g^2_1(y^2_1,\pi_1) = \left\{\begin{array}{ll}\text{1} & \text{if } y^2_1 = 0,\\\text{0} & \text{if } y^2_1 = 1.  \end{array}\right.}
and
\beqq{g^1_2(x_2,\pi_2) = 1 \qquad g^2_2(y^2_2,\pi_2) = 1. } 
% \beqq{g^i_2(y^i_2,\pi) =  \left\{\begin{array}{ll}\text{0} & \text{if } \pi = \pi^0,\\ 1 - x^2_1 & \text{if } \pi= \pi^1.  \end{array}\right. \qquad g^2_2(x^2_2,\pi) =  \left\{\begin{array}{ll}\text{0} & \text{if } \pi = \pi^0,\\ 1 - x^2_2 & \text{if } \pi= \pi^1.  \end{array}\right.}

\end{enumerate}

\section{Behavioral Strategies and Existence of Equilibrium}\label{sec:behave}
The results of Theorems \ref{thm:equlb_condition} and \ref{thm:backward_ind} provide sufficient conditions for a pair of strategies to be an equilibrium of game \textbf{G2}. Neither of these results addresses the question of existence of equilibrium. In particular, the result of Theorem~\ref{thm:backward_ind} states that the (pure strategy) Bayesian Nash equilibria of the one-stage Bayesian games $SG_t(\pi), t=T,\ldots,1$, may be used to find a Markov perfect equilibrium of game \textbf{G2} and hence a common information based Markov perfect equilibrium of \textbf{G1}. However, the games $SG_t(\pi)$ may not have any (pure strategy) Bayesian Nash equilibrium. 

As is common in finite games, we need to allow for behavioral strategies in order to ensure the existence of equilibria. Toward that end, we now reconsider the model of game \textbf{G1}. At each time $t$, each controller  is now allowed to select a probability distribution $\VEC D^i_t$ over the (finite) set of actions $\mathcal{U}^i_t, i=1,2$ according to a control law of the form:
\begin{equation}
\VEC D^i_t =  g^i_t(\VEC P^i_t, \VEC C_t).
\end{equation} 
The rest of the model is the same as in Section~\ref{sec:model}.  We denote the set of probability distributions over $\mathcal{U}^i_t$ by $\Delta(\mathcal{U}^i_t)$.  

Following exactly the same arguments as in Section~\ref{sec:virtual}, we can define an equivalent game  where virtual players select prescriptions that are functions from the set of private information $\mathcal{P}^i_t$ to the set $\Delta(\mathcal{U}^i_t)$ and establish the result of Theorem \ref{thm:equiv} for this case. A sufficient condition for Markov perfect equilibrium of this game is given by  Theorem \ref{thm:equlb_condition} where $\gamma^i$ are now interpreted as mappings from $\mathcal{P}^i_t$ to $\Delta(\mathcal{U}^i_t)$ (instead of mappings from $\mathcal{P}^i_t$ to $\mathcal{U}^i_t$). Given a Markov perfect equilibrium $(\psi^1,\psi^2)$ of the virtual players' game, the equivalent  strategies $g^i_t(\cdot, \pi):=\psi^i_t(\pi)$ form a common information based Markov perfect equilibrium of game \textbf{G1} in behavioral strategies.

Further, we can follow a backward induction procedure identical to the one used in section~\ref{sec:backward_ind} (Algorithm 1), but now consider mixed strategy Bayesian Nash equilibria of the one-stage Bayesian games $SG_t(\pi)$ constructed there. We proceed as follows:\\
\underline{\textbf{Algorithm 2:}}
\begin{enumerate}
\item At the terminal time $T$, for each realization $\pi$ of the common information based belief at time $T$, consider the  one-stage Bayesian game  $SG_T(\pi)$ defined in Algorithm 1. A mixed strategy $\gamma^i$ for the game $SG_T(\pi)$ is a mapping form $\mathcal{P}^i_T$ to $\Delta(\mathcal{U}^i_T)$. A mixed strategy Bayesian Nash equilibrium of this game is a pair of strategies $\gamma^1,\gamma^2$  such that for any realization $\VEC p^i$, $\gamma^i(\VEC p^i)$  assigns zero probability to any action that is not a solution of the minimization problem 
\[ \min_{\VEC u^i}  \mathds{E}[c^i(\VEC X_t,\VEC u^i, \VEC U^{j}_t) |\VEC P^i_T = \VEC p^i ], \]
where $\VEC U^j_t$ is distributed according to $\gamma^{j}(\VEC P^{j}_t)$. Since $SG_t(\pi)$ is a finite Bayesian game, a mixed strategy equilibrium is guaranteed to exist \cite{Myerson_gametheory}. For any mixed strategy Bayesian Nash equilibrium $\gamma^{1*},\gamma^{2*}$ of $SG_T(\pi)$, denote the expected equilibrium costs as $V^i_T(\pi)$ and define $\psi^i_t(\pi) := \gamma^{i*}$, $i=1,2$. 

\item At time $t < T$, for each realization $\pi$ of the common information based belief at time $t$, consider the one-stage Bayesian  game $SG_t(\pi)$ defined in Algorithm 1.  A mixed strategy Bayesian Nash equilibrium of this game is a pair of strategies $\gamma^1,\gamma^2$  such that for any realization $\VEC p^i$, $\gamma^i(\VEC p^i)$  assigns zero probability to any action that is not a solution of the minimization problem 
\[ \min_{\VEC u^i}  \mathds{E}[c^i(\VEC X_t,\VEC u^i, \VEC U^{j}_t)) +V^i_{t+1}(F_t(\pi, \VEC Z_{t+1})) |\VEC P^i_t = \VEC p^i ], \]
where $\VEC U^j_t$ is distributed according to $\gamma^{j}(\VEC P^{j}_t)$ and $\VEC Z_{t+1}$ is the increment in common information generated according to \eqref{eq:commoninfo}, \eqref{eq:observation} and \eqref{eq:state} when control actions $\VEC U^i_t =\VEC u^i$ and $\VEC U^j_t$ distributed according to $\gamma^{j}(\VEC P^{j}_t)$ are used. Since $SG_t(\pi)$ is a finite Bayesian game, a mixed strategy equilibrium is guaranteed to exist \cite{Myerson_gametheory}.
For any mixed strategy Bayesian Nash equilibrium $\gamma^{1*},\gamma^{2*}$ of $SG_t(\pi)$, denote the expected equilibrium costs as $V^i_t(\pi)$ and define $\psi^i_t(\pi) := \gamma^{i*}$, $i=1,2$.

%Consider any Bayesian Nash equilibrium $\gamma^{1*},\gamma^{2*}$ of $SG_t(\pi)$ and denote the equilibrium costs as $V^i_t(\pi)$. Define $\psi^i_t(\pi) := \gamma^{i*}$, $i=1,2$. 
\end{enumerate}

  We can now state the following theorem.
\begin{theorem}
%\begin{enumerate}
 For the finite game \textbf{G1}, a common information based Markov  perfect equilibrium in behavioral strategies always exists.  
 Further, this equilibrium can be found by first constructing  strategies $\psi^1,\psi^2$ according to the  backward inductive procedure of Algorithm 2  and then defining behavioral strategies $\VEC g^1, \VEC g^2$ in \textbf{G1} as
\[ g^i_t(\cdot, \pi_t) := \psi^i_t(\pi_t), \]
$i=1,2$, $t=1,2,\ldots,T$. 
%\end{enumerate}
\end{theorem}
%\begin{theorem}
% The strategies $\psi^1,\psi^2$ defined by the above backward inductive procedure form a Markov perfect equilibrium of game \textbf{G2^b}. Consequently, behavioral strategies $\VEC g^1, \VEC g^2$ defined as
%\[ g^i_t(\cdot, \pi_t) := \psi^i_t(\pi_t), \]
%$i=1,2$, $t=1,2,\ldots,T$ form a common information based Markov perfect equilibrium of game \textbf{G1}. Moreover, For the finite game \textbf{G1}, a common information based Markov  perfect equilibrium in behavioral strategies always exists. Further, such an equilibrium
%\end{theorem}

%Recall that due to Lemma~\ref{lemma:evolution} (and in absence of Assumption~\ref{assm:separation}), the belief $\pi_t$ evolves as 
%\[ \pi_{t+1} = F_t(\pi_t, \gamma^1_t, \gamma^2_t, \VEC z_{t+1}). \]
%Thus, in order for a virtual player to be able to evaluate the current realization of $\pi_t$, it must know the prescriptions used by the other virtual player. However, the past prescriptions are not observed by virtual players since the only data they have available is $\VEC c_t$. Thus, a virtual player cannot evaluate the belief $\pi_t$ without knowing how the other player selects its prescriptions. 
%\input infinite_horizon
%\input lqg_games_v2
%\input lqg_onestep
%%%%%%%%%%%%%%%%%%%%%%%%%%%%%%%%%%%%%%%%%%%%%%%%%%%\input section6new_3%%%%%%%%%%%%%%%%%%%%%%%%%%%%%%%
\section{Discussion}\label{sec:comments}
\subsection{Importance of Assumption \ref{assm:separation}}
The most restrictive assumption in our analysis of game \textbf{G1} is  Assumption \ref{assm:separation} which states that the common information based belief is independent of control strategies. It is instructive to consider why our analysis does not work in the absence of this assumption. Let us  consider the model of Section~\ref{sec:model} with Assumption~\ref{assm:infoevolution} as before but without Assumption~\ref{assm:separation}. Lemma~\ref{lemma:evolution}, which follows from Assumption~\ref{assm:infoevolution}, is still true. For this version of game \textbf{G1} without Assumption~\ref{assm:separation}, we can construct an equivalent game with virtual players similar to game \textbf{G2}. Further, it is easy to show that Theorem~\ref{thm:equiv} which relates equilibria of \textbf{G2} to those of \textbf{G1} is still true. 

The key result for our analysis of game \textbf{G2} in section \ref{sec:virtual} was  Lemma~\ref{lemma:infostatelemma} which allowed us to use $\pi_t$ as a Markov state  and to define and characterize Markov perfect equilibria for the game \textbf{G2}.  Lemma~\ref{lemma:infostatelemma} essentially states that the set of Markov decision strategy pairs (that is, strategies that select prescriptions as a function of $\pi_t$) is closed with respect to the best response mapping. In other words, if we start with any pair of Markov strategies $(\psi^1,\psi^2)$ for the virtual players and define $\chi^i$ to be the best response of virtual player $i$ to $\psi^j$, then, for at least one choice of best response strategies, the pair $(\chi^1,\chi^2)$ belongs to the set of Markov strategy pairs. This is true not just for strategies $(\psi^1,\psi^2)$ that form an equilibrium but for any choice of Markov strategies. We will now argue that this is not necessarily true without Assumption~\ref{assm:separation}.

Recall that due to Lemma~\ref{lemma:evolution}, the belief $\pi_t$ evolves as 
\[ \pi_{t} = F_{t-1}(\pi_{t-1}, \gamma^1_{t-1}, \gamma^2_{t-1}, \VEC z_{t}). \]
Thus, in order  to evaluate the current realization of $\pi_t$, a virtual player must know the prescriptions used by both virtual players. However, the virtual players do not observe each other's past prescriptions since the only data they have available is $\VEC c_t$. Thus, a virtual player cannot evaluate the belief $\pi_t$ without knowing (or assuming) how the other player selects its prescriptions.  

Consider now  decision strategies $(\psi^1,\psi^2)$ for the two virtual players which operate as follows: At each time $t$, the prescriptions chosen by virtual players are 
\begin{equation} \label{eq:sec4.1}\gamma^i_t = \psi^i_t(\pi_t) \end{equation}
and the belief at the next time $t+1$ is 
\begin{equation}\label{eq:sec4.2}
\pi_{t+1} = F_t(\pi_t, \psi^1_t(\pi_t),\psi^2_t(\pi_t), \VEC z_{t+1}). \end{equation}
Assume that the above strategies are not a Nash equilibrium for the virtual players' game. Therefore, one virtual player, say virtual player 2, can benefit by deviating from its strategy. 
Given that virtual player $1$ continues to operate according to \eqref{eq:sec4.1} and \eqref{eq:sec4.2}, is it possible for virtual player $2$ to reduce its cost by using a non-Markov strategy, that is, a strategy that selects prescriptions based on more data than just $\pi_t$? Consider any time $t$, if virtual player $2$ has deviated to some other choice of Markov decision rules $\psi^{2*}_{1:t-1}$ in the past, then the \emph{true belief on state and private information given the common information}, 
\[\pi^{*}_t = \mathds{P}^{\psi^1_{1:t-1},\psi^{2*}_{1:t-1}}(\VEC x_t, \VEC p^1_t, \VEC p^2_t |\VEC c_t), \] is different from the belief $\pi_t$ evaluated by the first player according to \eqref{eq:sec4.2}. (Note that since past prescriptions are not observed and virtual player 1's operation is fixed by \eqref{eq:sec4.1} and \eqref{eq:sec4.2}, virtual player $1$ continues to use $\pi_t$ evolving according to \eqref{eq:sec4.2} as its belief.) Even though $\pi_t$ is no longer the true belief, virtual player 2 can still track its evolution using \eqref{eq:sec4.2}. Using arguments similar to those in the proofs of Lemmas~\ref{lemma:markovprop} and \ref{lemma:infostatelemma}, it can be established that an optimal strategy for virtual player $2$, given that virtual player $1$ operates according to  \eqref{eq:sec4.1} and \eqref{eq:sec4.2}, is of the form $\gamma^2_t = \psi^{2*}_t (\pi^{*}_t,\pi_t)$, where $\pi^*_t$ is the true conditional belief on state and private information given the common information whereas $\pi_t$ is given by \eqref{eq:sec4.2}. Thus, the best response of player $2$ may not necessarily be a Markov strategy and hence Lemma~\ref{lemma:infostatelemma} may no longer hold. Without Lemma~\ref{lemma:infostatelemma}, we cannot define  Markov perfect equilibrium of game \textbf{G2} using $\pi_t$ as the state.

\subsection{The Case of Team Problems}
The game \textbf{G1} is referred to as a \emph{team problem} if the two controllers have the same cost functions, that is, $c^1(\cdot)=c^2(\cdot)=c^{team}(\cdot)$. Nash equilibrium strategies can then be interpreted as person-by-person optimal strategies \cite{Ho:1980}. Clearly, the results of sections~\ref{sec:virtual} and \ref{sec:behave} apply to person-by-person optimal strategies for team problems as well. 
\par
For team problems, our results can be strengthened in two ways. Firstly, we can find \emph{globally optimal} strategies for the controllers in the team using the virtual player approach and secondly, we no longer need to make Assumption~\ref{assm:separation}.  Let us retrace our steps in section~\ref{sec:virtual} for the team problem without Assumption~\ref{assm:separation}:
\begin{enumerate}
\item We can once again introduce virtual players that observe the common information and select prescriptions for the controllers. The two virtual players have the same cost function. So game \textbf{G2} is now a team problem and we will refer to it as \textbf{T2} . It is straightforward to establish that globally optimal strategies for virtual player can be translated to globally optimal strategies for the controllers in the team in a manner identical to Theorem~\ref{thm:equiv}.
\item Since we are no longer making Assumption~\ref{assm:separation}, the common information belief evolves according to 
\begin{equation} \pi_{t} = F_{t-1}(\pi_{t-1}, \gamma^1_{t-1}, \gamma^2_{t-1}, \VEC z_{t}). \label{eq:sec6eq2}
\end{equation}
Virtual player $1$ does not observe $\gamma^2_{t-1}$, so it cannot carry out the update described in \eqref{eq:sec6eq2}. \emph{However, we will now increase the information available to virtual players and assume that each virtual player can indeed observe all past prescriptions $\gamma^1_{1:t-1},\gamma^2_{1:t-1}$}. We refer to this team with expanded information for the virtual players as $\textbf{T2'}$. 

It should be noted that the globally optimal expected cost for \textbf{T2'} can be no larger than the globally optimal cost of \textbf{T2} since we have only added information in going from \textbf{T2} to \textbf{T2'}. We will later show that the globally optimal strategies we find for \textbf{T2'}  can be translated to equivalent strategies for \textbf{T2} with the same expected cost. %Thus, finding globally optimal strategies in \textbf{T2'} will lead us to find globally optimal strategies for \textbf{T2}.
\item For \textbf{T2'}, since all past prescriptions are observed, both virtual players can evaluate $\pi_t$ using \eqref{eq:sec6eq2} \emph{without knowing the past decision rules $\psi^1_{1:t-1},\psi^2_{1:t-1}$}. We can now repeat the arguments in the proof of  Lemma~\ref{lemma:markovprop} to show that an analogous result is true for team \textbf{T2'} as well. The team problem for the virtual players is now a Markov decision problem with $\pi_t$ evolving according to \eqref{eq:sec6eq2} as the Markov state and the prescription pair $(\gamma^1_t,\gamma^2_t)$ as the decision.  We can then write a dynamic program for this Markov decision problem. 
\begin{theorem}\label{thm:teamDP}
For the team problem \textbf{T2'} with virtual players, for each realization of $\pi_t$, the optimal prescriptions are the minimizers in the following dynamic program: 
\begin{align}
  &V^{team}_T(\pi) := \min_{\tilde{\gamma}^1,\tilde{\gamma}^2} \mathds{E}[c^{team}(\VEC X_t, {\Gamma}^1_T(\VEC P^1_T), \Gamma^{2}_T(\VEC P^{2}_T))|\Pi_T=\pi,\Gamma^1_T = \tilde{\gamma}^1, \Gamma^2_T = \tilde{\gamma}^2]
\end{align}
 
\begin{align}
  &V^{team}_t(\pi) := \min_{\tilde{\gamma}^1, \tilde{\gamma}^2}\mathds{E}[c^{team}(\VEC X_t, \Gamma^1_t(\VEC P^1_t), \Gamma^{2}_t(\VEC P^{2}_t)) +  V^{team}_{t+1}(\Pi_{t+1})|\Pi_t=\pi, \Gamma^1_t = \tilde{\gamma}^1, \Gamma^2_t = \tilde{\gamma}^2]
  \end{align}
  where  $\Pi_{t+1} = F_t(\Pi_t,\Gamma^1_t,\Gamma^2_t,\VEC Z_{t+1})$. %and $\VEC Z_{t+1}$ is the increment in common information generated according to \eqref{eq:commoninfo}, \eqref{eq:observation} and \eqref{eq:state} when control actions $\VEC U^i_t =\tilde{\gamma}^1(\VEC P^1_t)$ and $\VEC U^2_t = \gamma^{2}_t(\VEC P^{2}_t)$ are used.
  \end{theorem}
\item Let $\psi^{1*}_t(\pi)$ be the minimizer in the right hand side of  the definition of $V^{team}_t(\pi)$ in the above dynamic program. The globally optimal virtual players' operation can be described as: At each $t$, evaluate 
\begin{equation}
\pi_t = F_{t-1}(\pi_{t-1}, \gamma^1_{t-1}, \gamma^2_{t-1}, \VEC z_{t})\label{eq:team1}
\end{equation}
 and then select the prescriptions 
 \begin{equation}
 \gamma^i_t = \psi^{*i}_t(\pi_t) \quad i=1,2. \label{eq:team2}
 \end{equation}
 Now, instead of operating according to \eqref{eq:team1} and \eqref{eq:team2}, assume that virtual players operate as follows:
  At each $t$, evaluate 
\begin{equation}
\pi_t = F_{t-1}(\pi_{t-1}, \psi^{*1}_{t-1}(\pi_{t-1}),\psi^{*2}_{t-1}(\pi_{t-1}), \VEC z_{t})\label{eq:team3}
\end{equation}
 and then select the prescriptions 
 \begin{equation}
 \gamma^i_t = \psi^{*i}_t(\pi_t) ~~~~~ i=1,2. \label{eq:team4}
 \end{equation}
 It should be clear that virtual players operating according to \eqref{eq:team3} and \eqref{eq:team4} will achieve the same globally optimal performance as the virtual players operating according to \eqref{eq:team1} and \eqref{eq:team2}. Furthermore, the virtual players in \textbf{T2} can follow \eqref{eq:team3} and \eqref{eq:team4} and thus achieve the same globally optimal performance as in \textbf{T2'}.
 \end{enumerate}
 Thus, to find globally optimal strategies for the team of virtual players in absence of Assumption~\ref{assm:separation}, we first increased their information to include past prescriptions and then mapped the globally optimal strategies with increased information to equivalent strategies with original information. 
 \par
 For the \emph{game} \textbf{G2} in absence of assumption \ref{assm:separation}, we cannot follow the above approach of first increasing virtual players' information to include past prescriptions, finding equilibrium with added information and then mapping the equilibrium strategies to equivalent strategies with original information. 
 %Why could we not do a similar expansion of virtual players' information to find equilibrium for the game \textbf{G2} in absence of assumption \ref{assm:separation} by first finding equilibrium with expanded information?
  To see the reason, let us denote the virtual player operation given by \eqref{eq:team1} and \eqref{eq:team2} by the strategy $\sigma^i,i=1,2$ and the  virtual player operation given by \eqref{eq:team3} and \eqref{eq:team4} by the strategy $\hat{\sigma}^i,i=1,2$. Then, while it is true that
 $ \mathcal{J}^i(\sigma^1,\sigma^2) = \mathcal{J}^i(\hat{\sigma}^1,\hat{\sigma}^2), i=1,2,$
 but for some other strategies $\rho^1,\rho^2$, it is not necessarily true that 
 $ \mathcal{J}^i(\sigma^i,\rho^j) = \mathcal{J}^i(\hat{\sigma}^i,\rho^j), i,j=1,2,i\neq j$.
 Therefore, the equilibrium conditions for $\sigma^1,\sigma^2$:
 \begin{equation}
 \mathcal{J}^1(\sigma^1,\sigma^2) \leq   \mathcal{J}^1(\rho^1,\sigma^2), \quad \text{and}\quad \mathcal{J}^2(\sigma^1,\sigma^2) \leq  \mathcal{J}^2(\sigma^1,\rho^2),
\end{equation} 
do not necessarily imply the equilibrium conditions for $\hat{\sigma}^1,\hat{\sigma}^2$:
\begin{equation}
 \mathcal{J}^1(\hat{\sigma}^1,\hat{\sigma}^2) \leq   \mathcal{J}^1(\rho^1,\hat{\sigma}^2), \quad \text{and}\quad \mathcal{J}^2(\hat{\sigma}^1,\hat{\sigma}^2) \leq  \mathcal{J}^2(\hat{\sigma}^1,\rho^2).
 \end{equation}
 %where $\rho^1,\rho^2$ are any possible strategies for virtual players 1 and 2 respectively.
 \begin{Remark}
 Our dynamic program for the team problem is similar to the dynamic program for teams obtained in \cite{NMT2012} using a slightly different but conceptually similar approach.
 \end{Remark}

%%%%%%%%%%%%%%%%%%%%%%%%%%%%%%%%5%%%%%%%%%%%%%%%%%%%%%%%%%%%%%%%%%%%%%%%%%%%%%%%%%%%%%%%%%%%%%%%%%%%%%%%%%%%%%%%%
%%%%%%%%%%%%%%%%%%%%%%%%%%%%%%%%%%%%%%%%%\input conclusion%%%%%%%%%%%%%%%%%%%%%%%%%%%%%%%%%%%%%%%%%%%%%%%
\section{Concluding Remarks}\label{sec:conclusion}
We considered the problem of finding Nash equilibria of a general model of stochastic games with asymmetric information. Our analysis relied on the nature of common and private information among the controllers. Crucially, we assumed that the common information among controllers is increasing with time and that a common information based belief on the system state and private information is independent of control strategies. Under these assumptions, the game with asymmetric information is shown to be equivalent to another game with symmetric information for which we obtained a characterization of Markov perfect equilibria. This characterization allowed us to provide a  backward induction algorithm to find Nash equilibria of the original game. Each step of this algorithm involves finding Bayesian Nash equilibria of a one-stage Bayesian game. The class of Nash equilibria of the original game that can be characterized in this backward manner are named \emph{common information based Markov perfect equilibria}.
\par
The class of common information based Markov perfect equilibria for asymmetric information games bears conceptual similarities with Markov perfect equilibria of symmetric information games with perfect state observation. In symmetric information games with perfect state observation, a controller may be using past state information only because the other controller is using that information. Therefore, if one controller restricts to Markov strategies, the other controller can do the same.  This observation provides the justification for focusing only on Markov perfect equilibria for such games. Our results show that a similar observation can be made in our model of games with asymmetric information. A controller may be using the entire common information only because other controller is using that information. If one controller chooses to only use the common information based belief on the state and private information, the other controller can do the same. Thus, it is reasonable to focus on the class of common information based Markov perfect equilibria for our model of games with asymmetric information.
\par
Further, for zero-sum games,  the uniqueness of the value of the game implies that the equilibrium cost of a common information based Markov perfect equilibrium is the same as the equilibrium cost of any other Nash equilibrium \cite{Osborne}.
\par
For finite games, it is always possible to find pure strategy Nash equilibria (if they exist) by a brute force search of the set of possible strategy profiles. The number of strategy choices for controller $i$ are $|\mathcal{U}^i_1|^{|\mathcal{P}^i_1 \times \mathcal{C}_1|} \times \ldots \times |\mathcal{U}^i_T|^{|\mathcal{P}^i_T \times \mathcal{C}_T|} $. For simplicity, assume that the set  of possible realizations of  private information $\mathcal{P}^i_t$ does not change with time. 
However, because the common information is required to be increasing with time (see Assumption \ref{assm:infoevolution}), the  cardinality of the set possible realization of common information $\mathcal{C}_t$ is exponentially increasing with time. Thus, the number of possible control strategies exhibits a double exponential growth with time. 
\par
Algorithm 1 provides an alternative way for finding an equilibrium by solving a succession of one stage Bayesian games. But how many such games need to solved? At each time $t$, we need to solve a Bayesian game for each possible realization of the belief $\pi_t$. Let $\mathcal{R}_t$ denote the set of possible realizations of the belief $\pi_t$. Since the belief is simply a function of the common information, we must have that $|\mathcal{R}_t| \leq |\mathcal{C}_t|$. Thus, the total number of one stage games that need to solved is no larger that $\sum_{t=1}^T |\mathcal{C}_t|$. Recalling the exponential growth of $|\mathcal{C}_t|$, the number of one-stage games to solve shows an exponential growth with time. This is clearly better than the double exponential growth for the brute force search.
\par
Two possible reasons may further reduce the complexity of Algorithm 1. Firstly, the set $|\mathcal{R}_t|$ may not be growing exponentially with time (as in the case of the information structure in Section~\ref{sec:Num_example}, where $|\mathcal{R}_t|= 3$, for all $t>1$). Secondly, the one-stage games at time $t$, $SG_t(\pi)$ may possess enough structure that it is possible to find an equilibrium for a generic $\pi$ that can be used to construct equilibrium for all choices of $\pi$. For finite games, it is not clear what additional features need to be present in game \textbf{G1} such that the resulting one-stage games $SG_t(\pi)$ can be solved for a generic $\pi$. In the sequel to this paper we will extend the approach used here to linear quadratic Gaussian games  and show that in these games it is possible to solve the one-stage games for a generic belief $\pi$.
\par
Conceptually, the approach adopted in this paper can be extended to infinite time horizon games with discounted costs under suitable stationarity conditions. However, in infinite horizon games, the number of possible realizations of the common information based belief would, in general, be infinite. Establishing the existence of common information based Markov perfect equilibria for infinite horizon games would be an interesting direction for future work in this area.
%%%%%%%%%%%%%%%%%%%%%%%%%%%%%%%%%%%%%%%%%%%%%%%%%%%%%%%%%%%%%%%%%%%%%%%%%%%%%%%%%%%%%%%%%%%%%%%%%%%%%%%%%5

\section{Acknowledgments}
This work was supported in part by the AFOSR MURI Grant FA9550-10-1-0573. The second author thanks Bharti Center for Telecommunications, IIT Bombay for infrastructural support and Research Internship in Science and Engineering program of Indo-US Science and Technology Forum for supporting the visit to Indian Institute of Technology Bombay.

\appendices
\section{Proof of Lemma \ref{lemma:evolution}}\label{sec:update_func}

 Consider a realization
$\VEC c_{t}$ of the common information $\VEC C_{t}$ at time $t$. Let $\gamma^1_t, \gamma^2_t$ be the corresponding
realization of the partial functions of the control laws at time $t$, that is, $\gamma^i_t = g^i_t(\cdot,\VEC c_t)$.
Given the realization of the common information based belief $\pi_t$ and the partial functions $\gamma^1_t, \gamma^2_t$, we can find the joint conditional distribution on $(\VEC X_t, \VEC P^1_t, \VEC P^2_t, \VEC X_{t+1}, \VEC P^1_{t+1}, \VEC P^2_{t+1}, \VEC Z_{t+1})$ conditioned on the common information at time $t$ as follows:
\begin{align}
&\mathds{P}^{g^1_{1:t},g^2_{1:t}}(\VEC x_t, \VEC p^1_t, \VEC p^2_t, \VEC x_{t+1}, \VEC p^1_{t+1}, \VEC p^2_{t+1}, \VEC z_{t+1}|\VEC c_t) \notag\\
&= \sum_{\VEC y^1_{t+1},\VEC y^2_{t+1}, \VEC u^1_t, \VEC u^2_t}\mathds{P}^{g^1_{1:t},g^2_{1:t}}(\VEC x_t, \VEC p^1_t, \VEC p^2_t, \VEC x_{t+1}, \VEC p^1_{t+1}, \VEC p^2_{t+1}, \VEC z_{t+1}, \VEC y^1_{t+1},\VEC y^2_{t+1}, \VEC u^1_t, \VEC u^2_t|\VEC c_t) \notag 
\end{align}
\begin{align}
&=\sum_{\VEC y^1_{t+1},\VEC y^2_{t+1}, \VEC u^1_t, \VEC u^2_t}\mathds{1}_{\{\zeta_{t+1}(\VEC p^1_t, \VEC p^2_t, \VEC u^1_t, \VEC u^2_t, \VEC y^1_{t+1}, \VEC y^2_{t+1})=\VEC z_{t+1}\}} \mathds{1}_{\{\xi^1_{t+1}(\VEC p^1_t, \VEC u^1_t,  \VEC y^1_{t+1}) = \VEC p^1_{t+1}\}}\mathds{1}_{\{\xi^2_{t+1}(\VEC p^2_t, \VEC u^2_t,  \VEC y^2_{t+1}) = \VEC p^2_{t+1}\}} \notag \\
&\mathds{P}(\VEC y^1_{t+1},\VEC y^2_{t+1}|\VEC x_{t+1})\mathds{P}(\VEC x_{t+1}|\VEC x_t, \VEC u^1_t, \VEC u^2_t)\mathds{1}_{\{\gamma^1_t(\VEC p^1_t)=\VEC u^1_t\}}\mathds{1}_{\{\gamma^2_t(\VEC p^2_t)=\VEC u^2_t\}}\pi_t(\VEC x_t, \VEC p^1_t, \VEC p^2_t)\label{eq:appendixeq1}
\end{align}
Note that in addition to the arguments on the left side of conditioning in \eqref{eq:appendixeq1}, we only need $\pi_t$ and $\gamma^1_t,\gamma^2_t$ to evaluate the right hand side of \eqref{eq:appendixeq1}. That is,  the joint conditional distribution on $(\VEC X_t, \VEC P^1_t, \VEC P^2_t, \VEC X_{t+1}, \VEC P^1_{t+1}, \VEC P^2_{t+1}, \VEC Z_{t+1})$ depends only on $\pi_t$, $\gamma^1_t$ and $\gamma^2_t$  with no dependence on control strategies.
\par
We can now consider the common information based belief at time $t+1$,
\begin{align}
&\pi_{t+1}(\VEC x_{t+1}, \VEC p^1_{t+1}, \VEC p^2_{t+1}) = \mathds{P}(\VEC x_{t+1}, \VEC p^1_{t+1}, \VEC p^2_{t+1}|\VEC c_{t+1}) \notag \\
&= \mathds{P}(\VEC x_{t+1}, \VEC p^1_{t+1}, \VEC p^2_{t+1}|\VEC c_{t},\VEC z_{t+1}) \notag \\
&= \frac{\mathds{P}(\VEC x_{t+1}, \VEC p^1_{t+1}, \VEC p^2_{t+1},\VEC z_{t+1}|\VEC c_{t})}{\mathds{P}(\VEC z_{t+1}|\VEC c_{t})}\label{eq:appendixeq2}
\end{align}
 The numerator and denominator of \eqref{eq:appendixeq2} are both marginals of the probability in \eqref{eq:appendixeq1}. Using \eqref{eq:appendixeq1} in \eqref{eq:appendixeq2}, gives $\pi_{t+1}$ as a function of $\pi_t,\gamma^1_t,\gamma^2_t,\VEC z_{t+1}$.

%%%%%%%%%%%%%%%%%%%%%%%%%%%%%%%%%%%%%%%%%%%%%%%%%%%%%%%%%%%%%%%%%%%%%%%%%%%%%%%%%%5
\section{Proof of Lemma \ref{lemma:markovprop}}\label{sec:markovprop}
Consider a realization $\VEC c_t$ of common information at time $t$ and realizations $\pi_{1:t},\gamma^1_{1:t},\gamma^2_{1:t}$ of beliefs and prescriptions till time $t$. 
Because of \eqref{eq:evolution1} in Assumption \ref{assm:separation}, we have
\[ \Pi_{t+1} = F_t(\pi_t,\VEC Z_{t+1})\]
Hence, in order to establish the lemma, it suffices to show that 
\begin{align}
\mathds{P}(\VEC Z_{t+1}|\VEC c_t, \pi_{1:t},\gamma^1_{1:t},\gamma^2_{1:t}) = \mathds{P}(\VEC Z_{t+1}|\pi_{t},\gamma^1_{t},\gamma^2_{t})\label{eq:appeq2.0}
\end{align}
Recall that 
\begin{align}
&\VEC Z_{t+1} = \zeta_{t+1}(\VEC P^1_t, \VEC P^2_t, \VEC U^1_t, \VEC U^2_t, \VEC Y^1_{t+1}, \VEC Y^2_{t+1})\notag \\
&=\zeta_{t+1}(\VEC P^1_t, \VEC P^2_t, \gamma^1_t(\VEC P^1_t), \gamma^2_t(\VEC P^2_t), \VEC Y^1_{t+1}, \VEC Y^2_{t+1}) \label{eq:appeq2.1}
\end{align}
where we used the fact that the control actions are  simply the prescriptions evaluated at the private information.
Therefore,
\begin{align}
&\mathds{P}(\VEC Z_{t+1} =\VEC z|\VEC c_t, \pi_{1:t},\gamma^1_{1:t},\gamma^2_{1:t}) \notag \\
&= \sum_{\VEC x_t,\VEC x_{t+1},\VEC y^1_{t+1},\VEC y^2_{t+1}, \VEC p^1_t, \VEC p^2_t}\mathds{P}(\VEC Z_{t+1} =\VEC z,\VEC x_t,\VEC x_{t+1},\VEC y^1_{t+1},\VEC y^2_{t+1}, \VEC p^1_t, \VEC p^2_t|\VEC c_t, \pi_{1:t},\gamma^1_{1:t},\gamma^2_{1:t})\notag 
\end{align}
\begin{align}
&= \sum_{\VEC x_t,\VEC y^1_{t+1},\VEC y^2_{t+1}, \VEC p^1_t, \VEC p^2_t} \mathds{1}_{\{\zeta_{t+1}(\VEC p^1_t, \VEC p^2_t, \gamma^1_t(\VEC p^1_t), \gamma^2_t(\VEC p^2_t), \VEC y^1_{t+1}, \VEC y^2_{t+1})=z\}} \mathds{P}(\VEC y^1_{t+1},\VEC y^2_{t+1}|\VEC x_{t+1})\notag \\
&\times\mathds{P}(\VEC x_{t+1}|\VEC x_t, \gamma^1_t(\VEC p^1_t), \gamma^2_t(\VEC p^2_t)) \mathds{P}(\VEC x_t, \VEC p^1_t, \VEC p^2_t|\VEC c_t, \pi_{1:t},\gamma^1_{1:t},\gamma^2_{1:t}) \notag 
\end{align}
\begin{align}
&= \sum_{\VEC x_t,\VEC y^1_{t+1},\VEC y^2_{t+1}, \VEC p^1_t, \VEC p^2_t} \mathds{1}_{\{\zeta_{t+1}(\VEC p^1_t, \VEC p^2_t, \gamma^1_t(\VEC p^1_t), \gamma^2_t(\VEC p^2_t), \VEC y^1_{t+1}, \VEC y^2_{t+1})=z\}} \mathds{P}(\VEC y^1_{t+1},\VEC y^2_{t+1}|\VEC x_{t+1})\notag \\
&\times\mathds{P}(\VEC x_{t+1}|\VEC x_t, \gamma^1_t(\VEC p^1_t), \gamma^2_t(\VEC p^2_t)) \pi_t(\VEC x_t, \VEC p^1_t, \VEC p^2_t), \label{eq:appeq2.2}
\end{align}
where we used the fact that 
$\mathds{P}(\VEC x_t, \VEC p^1_t, \VEC p^2_t|\VEC c_t, \pi_{1:t},\gamma^1_{1:t},\gamma^2_{1:t}) = \mathds{P}(\VEC x_t, \VEC p^1_t, \VEC p^2_t|\VEC c_t)$, since $\pi_{1:t},\gamma^1_{1:t},\gamma^2_{1:t}$ are all functions of $\VEC c_t$, and the fact that $\mathds{P}(\VEC x_t, \VEC p^1_t, \VEC p^2_t|\VEC c_t)=: \pi_t(\VEC x_t, \VEC p^1_t, \VEC p^2_t)$.
The right hand side in \eqref{eq:appeq2.2} depends only on $\pi_t$ and $\gamma^1_t,\gamma^2_t$. Thus, the conditional probability of $\VEC Z_{t+1}=\VEC z$ conditioned on $\VEC c_t,\pi_{1:t},\gamma^1_{1:t},\gamma^2_{1:t}$ depends only on $\pi_t$ and $\gamma^1_t,\gamma^2_t$. This establishes \eqref{eq:appeq2.0} and hence the lemma.
%%%%%%%%%%%%%%%%%%%%%%%%%%%%%%%%%%%%%%%%%%%%%%%%%%%%%%%%%%%%%%%%%%%%%%%%%%%%%%%%%%%%%%%%%%%%%%%%%%
\section{Proof of Lemma \ref{lemma:infostatelemma}}\label{sec:infostatelemma}
Assume that virtual player $1$ is using a fixed strategy of the form $ \Gamma^1_t = \psi^1_t(\Pi_t)$, $t=1,2,\ldots,T$. We now want to find a strategy of virtual player 2 that is a best response to the given strategy of virtual player 1. Lemma \ref{lemma:markovprop} established that $\Pi_t$ is a controlled Markov process with the prescriptions $\Gamma^1_t, \Gamma^2_t$ as the controlling actions. Since $\Gamma^1_t$ has been fixed to $\psi^1_t(\Pi_t)$, it follows that, under the fixed strategy of virtual player 1, $\Pi_t$ can be viewed as a controlled Markov process with the decisions of virtual player 2, $\Gamma^2_t$ as the controlling action. 

At time $t$, if $\VEC c_t$ is the realization of common information, $\pi_t$ is the corresponding realization of the common information belief, then $\gamma^1_t = \psi^1_t(\pi_t)$ is prescription selected by virtual player 1. If virtual player 2 selects $\gamma^2_t$, the expected instantaneous cost for the  virtual player 2 is
\begin{align}
&\mathds{E}[c^2(\VEC X_t, \VEC U^1_t, \VEC U^2_t)|\VEC c_t] = \mathds{E}[c^2(\VEC X_t, \gamma^1_t(\VEC P^1_t), \gamma^2_t(\VEC P^2_t))|\VEC c_t]\notag \\
&= \sum_{\VEC x_t, \VEC p^1_t, \VEC p^2_t}c^2(\VEC x_t, \gamma^1_t(\VEC p^1_t), \gamma^2_t(\VEC p^2_t))\mathds{P}(\VEC x_t, \VEC p^1_t, \VEC p^2_t|\VEC c_t)\notag\\
&= \sum_{\VEC x_t, \VEC p^1_t, \VEC p^2_t}c^2(\VEC x_t, \gamma^1_t(\VEC p^1_t), \gamma^2_t(\VEC p^2_t))\pi_t(\VEC x_t, \VEC p^1_t, \VEC p^2_t) =: \tilde{c}^2(\pi_t,\gamma^2_t)
\end{align}
Thus, given the fixed strategy of virtual player 1, the instantaneous expected cost for virtual player 2 depends only on the belief $\pi_t$ and the prescription selected by virtual player 2. Given the controlled Markov nature of $\pi_t$, it follows that virtual player 2's optimization problem is a Markov decision problem with $\Pi_t$ as the state and hence virtual player 2 can optimal select its prescription as a function of $\Pi_t$. This completes the proof of the lemma.

%%%%%%%%%%%%%%%%%%%%%%%%%%%%%%%%%%%%%%%%%%%%%%%%%%%%%%%%%%%%%%%%%%%%%%%%%%%%%%%%%%%%%%%%%%%%%%%%%%%%
\section{Proof of Theorem \ref{thm:equlb_condition}} \label{sec:equlb_condition}
Consider a strategy pair $(\psi^1,\psi^2)$ that satisfies the conditions of the theorem. For any $1 \leq k \leq T$ and any realization $\VEC c_k$ of the common information at time $k$, we want to show that the strategies form a Nash equilibrium of the sub-game starting from time $k$ with the costs given as
\begin{align}\label{eq:appD.1}
 &\mathds{E}\Big[ \sum_{t=k}^T c^i(\VEC X_t, \VEC U^1_t, \VEC U^2_t) | \VEC c_k \Big],
 \end{align}
$i=1,2$. If the strategy of player $j$ is fixed to $\psi^j_t, t=k,k+1,\ldots,T$, then by arguments similar to those in the proof of Lemma \ref{lemma:infostatelemma}, the optimization problem for player $i$ starting from time $k$ onwards with the objective given by \eqref{eq:appD.1} is a Markov decision problem which we denote by $MDP^i_k$. Since $\psi^i_t, t=k,k+1,\ldots,T,$ satisfy the conditions of Theorem \ref{thm:equlb_condition} for player $i$, they satisfy the dynamic programming conditions of $MDP^i_k$. Thus, $\psi^i_t,t=k,k+1,\ldots,T,$ is the best response to $\psi^j,t=k,k+1,\ldots,T,$ in the sub-game starting from time $k$. Interchanging the roles of $i$ and $j$ implies that the strategies $\psi^1_t,\psi^2_t,t=k,k+1,\ldots,T,$ form an equilibrium of the sub-game  starting from time $k$. Since $k$ was arbitrary, this completes the proof of sufficiency part of the theorem. The converse follows a similar MDP based argument.  

%%%%%%%%%%%%%%%%%%%%%%%%%%%%%%%%%%%%%%%%%%%%%%%%%%%%%%%%%%%%%%%%%%%%%%%%%%%%%%%%%%%%%%%%%%%%%%%%
\section{Proof of Theorem \ref{thm:backward_ind}} \label{sec:backward_indproof}

Consider any realization $\pi$ of the common information based belief and consider a Bayesian Nash equilibrium $\gamma^{1*},\gamma^{2*}$  of the game $SG_T(\pi)$. We will show that $\gamma^{1*},\gamma^{2*}$ satisfy the value function conditions for time $T$ in Theorem \ref{thm:equlb_condition}. By definition of Bayesian Nash equilibrium, for every realization $\VEC p^1$ of $\VEC P^1_T$,
\begin{align}
& \mathds{E}^{\pi}[c^1(\VEC X_T,\gamma^{1*}(\VEC P^1_T), \gamma^{2*}(\VEC P^{2}_T))|\VEC P^1_T = \VEC p^1] \leq \mathds{E}^{\pi}[c^1(\VEC X_T,\tilde{\gamma}^{1}(\VEC P^1_T), \gamma^{2*}(\VEC P^{2}_T))|\VEC P^1_T = \VEC p^1],
\end{align}
%\begin{align}
%&\mathds{E}[c^1(\VEC X_T,\gamma^{1*}(\VEC p^1), \gamma^{2*}(\VEC P^{2}_T))|\VEC P^1_T = \VEC p^1] \leq \mathds{E}[c^1(\VEC X_T,\tilde{\gamma}^{1}(\VEC p^1), \gamma^{2*}(\VEC P^{2}_T))|\VEC P^1_T = \VEC p^1], 
%\end{align}
for any choice of $\tilde{\gamma}^1$. Averaging over $\VEC p^1$, we get
\begin{align}
&\mathds{E}^{\pi}\Big[\mathds{E}[c^1(\VEC X_T,\gamma^{1*}(\VEC P^1_T), \gamma^{2*}(\VEC P^{2}_T))|\VEC P^1_T]\Big] \leq \mathds{E}^{\pi}\Big[\mathds{E}[c^1(\VEC X_T,\tilde{\gamma}^{1}(\VEC P^{1}_T), \gamma^{2*}(\VEC P^{2}_T))|\VEC P^1_T]\Big] \notag \\
& \implies \mathds{E}^{\pi}[c^1(\VEC X_T,\gamma^{1*}(\VEC P^{1}_T), \gamma^{2*}(\VEC P^{2}_T))] \leq \mathds{E}^{\pi}[c^1(\VEC X_T,\tilde{\gamma}^{1}(\VEC P^{1}_T), \gamma^{2*}(\VEC P^{2}_T))], \label{eq:appE.1}
\end{align}
where all the expectations are with respect to the belief $\pi$ on $(\VEC X_T,\VEC P^1_T, \VEC P^2_T)$.
Similarly, 
\begin{align}
\mathds{E}^{\pi}[c^2(\VEC X_T,\gamma^{1*}(\VEC P^{1}_T), \gamma^{2*}(\VEC P^{2}_T))] \leq \mathds{E}^{\pi}[c^2(\VEC X_T,{\gamma}^{1*}(\VEC P^{1}_T), \tilde{\gamma}^{2}(\VEC P^{2}_T))],
\end{align}
for any choice of $\tilde{\gamma}^2$. Thus, $\psi^i_T(\pi) := \gamma^{i*}$, $i=1,2$ satisfy the conditions in \eqref{eq:dpeq1} and \eqref{eq:dpeq2} when $\Pi_T=\pi$.

 Similarly, for any time $t<T$, consider  any realization $\pi$ of the common information based belief at $t$ and consider a Bayesian Nash equilibrium $\gamma^{1*},\gamma^{2*}$  of the game $SG_t(\pi)$. Then, by definition of Bayesian Nash equilibrium, for every realization $\VEC p^1$ and any choice of $\tilde{\gamma}^1$, we have that the expression 
\[\mathds{E}^{\pi}[c^1(\VEC X_t,\gamma^{1*}(\VEC P^i_t), \gamma^{2*}(\VEC P^{2}_t)) +V^1_{t+1}(F_t(\pi, \VEC Z_{t+1})) |\VEC P^1_t = \VEC p^1],\]
(where $\VEC Z_{t+1}$ is the increment in common information generated according to \eqref{eq:commoninfo}, \eqref{eq:observation} and \eqref{eq:state} when control actions $\VEC U^1_t =\gamma^{1*}(\VEC p^i)$ and $\VEC U^2_t = \gamma^{2*}(\VEC P^{2*}_t)$ are used) can be no larger than
\[\mathds{E}^{\pi}[c^1(\VEC X_t,\tilde{\gamma}^{1}(\VEC P^i_t), \gamma^{2*}(\VEC P^{2}_t)) +V^1_{t+1}(F_t(\pi, \VEC Z_{t+1})) |\VEC P^1_t = \VEC p^1],\]
(where $\VEC Z_{t+1}$ is the increment in common information generated according to \eqref{eq:commoninfo}, \eqref{eq:observation} and \eqref{eq:state} when control actions $\VEC U^1_t =\tilde{\gamma}^{1}(\VEC p^i)$ and $\VEC U^2_t = \gamma^{2*}(\VEC P^{2*}_t)$ are used. Similar conditions hold for player 2.  Averaging over $\VEC p^1,\VEC p^2,$ establishes that $\psi^i_t(\pi) := \gamma^{i*}$, $i=1,2$ satisfy the conditions in \eqref{eq:dpeq3} and \eqref{eq:dpeq4} when $\Pi_t=\pi$. 
\par
Thus,  the strategies $\psi^i,i=1,2$ defined by the backward induction procedure of Algorithm 1  satisfy the  conditions of Theorem \ref{thm:equlb_condition} and hence form a Markov perfect equilibrium for game \textbf{G2}.% Therefore, $g^i_t(\cdot, \pi_t) := \psi^i_t(\pi_t)$ form a common information based Markov perfect equilibrium of \textbf{G1}
%Since \eqref{eq:appE.1} is true for any $\tilde{\gamma}^1$, it  implies that $\gamma^{1*}$ is a minimizer of the right hand side of \eqref{eq:dpeq1}. Similar arguments hold for $\gamma^{2*}$.

%%%%%%%%%%%%%%%%%%%%%%%%%%%%%%%%%%%%%%%%%%%%%%%%%%%%%%%%%%%%%%%%%%%%%%%%%%%%%%%%%%%%%%%%%%%
%\input appendix_part2
%\input lqg_appendix

\bibliographystyle{IEEEtran}
\bibliography{IEEEabrv,myref,collection}

\begin{thebibliography}{10}
\providecommand{\url}[1]{#1}
\csname url@samestyle\endcsname
\providecommand{\newblock}{\relax}
\providecommand{\bibinfo}[2]{#2}
\providecommand{\BIBentrySTDinterwordspacing}{\spaceskip=0pt\relax}
\providecommand{\BIBentryALTinterwordstretchfactor}{4}
\providecommand{\BIBentryALTinterwordspacing}{\spaceskip=\fontdimen2\font plus
\BIBentryALTinterwordstretchfactor\fontdimen3\font minus
  \fontdimen4\font\relax}
\providecommand{\BIBforeignlanguage}[2]{{%
\expandafter\ifx\csname l@#1\endcsname\relax
\typeout{** WARNING: IEEEtran.bst: No hyphenation pattern has been}%
\typeout{** loaded for the language `#1'. Using the pattern for}%
\typeout{** the default language instead.}%
\else
\language=\csname l@#1\endcsname
\fi
#2}}
\providecommand{\BIBdecl}{\relax}
\BIBdecl

\bibitem{Shapley:1953}
L.~S. Shapley, ``Stochastic games,'' \emph{Proc. Natl. Acad. Sci. USA},
  vol.~39, pp. 1095--1100, 1953.

\bibitem{Sobel1971}
\BIBentryALTinterwordspacing
M.~J. Sobel, ``\BIBforeignlanguage{English}{Noncooperative stochastic games},''
  \emph{\BIBforeignlanguage{English}{The Annals of Mathematical Statistics}},
  vol.~42, no.~6, pp. 1930--1935, 1971. [Online]. Available:
  \url{http://www.jstor.org/stable/2240119}
\BIBentrySTDinterwordspacing

\bibitem{Tirole}
D.~Fudenberg and J.~Tirole, \emph{Game Theory}.\hskip 1em plus 0.5em minus
  0.4em\relax MIT Press, 1991.

\bibitem{Basarbook}
T.~Ba\c{s}ar and G.~J. Olsder, \emph{Dynamic Non-cooperative Game
  Theory}.\hskip 1em plus 0.5em minus 0.4em\relax SIAM Series in Classics in
  Applied Mathematics, Philadelphia, 1999.

\bibitem{Filarbook}
J.~Filar and K.~Vrieze, \emph{Competitive Markov Decision Processes}.\hskip 1em
  plus 0.5em minus 0.4em\relax Springer, 1996.

\bibitem{Behn68}
R.~Behn and Y.-C. Ho, ``On a class of linear stochastic differential games,''
  \emph{IEEE Trans. Autom. Contr.}, vol.~13, no.~3, pp. 227 -- 240, Jun 1968.

\bibitem{Rhodes69}
I.~Rhodes and D.~Luenberger, ``Differential games with imperfect state
  information,'' \emph{IEEE Trans. Autom. Contr.}, vol.~14, no.~1, pp. 29 --
  38, Feb 1969.

\bibitem{Willman69}
W.~Willman, ``Formal solutions for a class of stochastic pursuit-evasion
  games,'' \emph{IEEE Trans. Autom. Contr.}, vol.~14, no.~5, pp. 504 -- 509,
  Oct 1969.

\bibitem{Hojota}
{Y. C. Ho}, ``On the minimax principle and zero-sum stochastic differential
  games,'' \emph{Journal of Optimization Theory and Applications}, vol.~13,
  no.~3, pp. 343--361, 1974.

\bibitem{MintzBasar}
T.~Ba\c{s}ar and M.~Mintz, ``A multistage pursuit-evasion game that admits a
  {Gaussian} random process as a maximin control policy,'' \emph{Stochastics},
  vol. 1:1-4, pp. 25--69, 1973.

\bibitem{Basaronestep}
T.~Ba\c{s}ar, ``Two-criteria {LQG} decision problems with one-step delay
  observation sharing pattern,'' \emph{Information and Control}, vol.~38, pp.
  21--50, 1978.

\bibitem{Basarmulti}
------, ``Decentralized multicriteria optimization of linear stochastic
  systems,'' \emph{IEEE Trans. Autom. Contr.}, vol.~23, no.~2, pp. 233 -- 243,
  Apr. 1978.

\bibitem{Altman:2009}
E.~Altman, V.~Kambley, and A.~Silva, ``Stochastic games with one step delay
  sharing information pattern with application to power control,'' in
  \emph{{P}roceedings of International Conference on Game Theory for Networks,
  GameNets'09}, May 2009, pp. 124--129.

\bibitem{hespanha2001}
J.~Hespanha and M.~Prandini, ``Nash equilibria in partial-information games on
  {M}arkov chains,'' in \emph{Proc. of the 40th IEEE Conference on Decision and
  Control}, 2001, pp. 2102--2107.

\bibitem{Bjota}
{T. Ba\c{s}ar}, ``On the saddle-point solution of a class of stochastic
  differential games,'' \emph{Journal of Optimization Theory and Applications},
  vol.~33, no.~4, pp. 539--556, 1981.

\bibitem{Cole2001}
H.~Cole and N.~Kocherlakota, ``Dynamic games with hidden actions and hidden
  states,'' \emph{Journal of Economic Theory}, vol.~98, no.~1, pp. 114--126,
  2001.

\bibitem{NMT:2011}
A.~Nayyar, A.~Mahajan, and D.~Teneketzis, ``Optimal control strategies in
  delayed sharing information structures,'' \emph{IEEE Transactions on
  Automatic Control}, vol.~57, no.~7, pp. 1606--1620, July 2011.

\bibitem{NayyarBasarcdc}
A.~Nayyar and T.~Ba\c{s}ar, ``Dynamic stochastic games with asymmetric
  information,'' accepted in \emph{51st IEEE Conference on Decision and
  Control, 2012}.

\bibitem{KumarVaraiya}
P.~R. Kumar and P.~Varaiya, \emph{Stochastic Systems: Estimation,
  Identification and Adaptive Control}.\hskip 1em plus 0.5em minus 0.4em\relax
  Prentice Hall, Englewood Cliffs, NJ, 1986.

\bibitem{MaskinTirole}
\BIBentryALTinterwordspacing
E.~Maskin and J.~Tirole, ``Markov perfect equilibrium: I. observable actions,''
  \emph{Journal of Economic Theory}, vol. 100, no.~2, pp. 191 -- 219, 2001.
  [Online]. Available:
  \url{http://www.sciencedirect.com/science/article/pii/S0022053100927856}
\BIBentrySTDinterwordspacing

\bibitem{Osborne}
M.~J. Osborne and A.~Rubinstein, \emph{A Course in Game Theory}.\hskip 1em plus
  0.5em minus 0.4em\relax MIT Press, 1994.

\bibitem{Myerson_gametheory}
R.~B. Myerson, \emph{Game Theory: Analysis of Conflict}.\hskip 1em plus 0.5em
  minus 0.4em\relax Harvard University Press, Cambridge, MA, 1997.

\bibitem{Ho:1980}
Y.-C. Ho, ``Team decision theory and information structures,'' \emph{Proc.
  IEEE}, vol.~68, no.~6, pp. 644--654, 1980.

\bibitem{NMT2012}
A.~Nayyar, A.~Mahajan, and D.~Teneketzis, ``Decentralized stochastic control
  with partial sharing information structures: A common information approach,''
  \emph{IEEE Transacions on Automatic Control}, Dec 2011, submitted.

\end{thebibliography}
\newpage
\vspace{20pt}
\begin{center}\textbf{\large{Supplementary Material}}\end{center}
%%%%%%%%%%%%%%%%%%%%%%%%%%%%%%%%%%%%%%%%%%%%%%%%%%%%%%%%%%%%%%%%%%%%%%%%%55\input supp_material
%%%%%%%%%%%%%%%%%%%%%%%%%%%%%%%%%%%%%%%%%%%%%%%%%%%%%%%%%%%%%%%%%%%%%%%%%%%%%%%%%%%%%%%%%%%%%%%%%%%%%%%%%
\section{Proof of Lemma \ref{lemma:onestep}}\label{sec:onestep}
\begin{proof}
 It is straightforward to verify that the structure of common and private information satisfies Assumption~\ref{assm:infoevolution}. We focus on the proof for Assumption~\ref{assm:separation}. For a realization $\VEC y^1_{1:t},\VEC y^2_{1:t}$,$\VEC u^1_{1:t},\VEC u^2_{1:t}$ of the common information at time $t+1$, the common information based belief  can be written as
\begin{align}
&\pi_{t+1}(\VEC x_{t+1}, \VEC y^1_{t+1}, \VEC y^2_{t+1}) = \mathds{P}^{g^1_{1:t},g^2_{1:t}}(\VEC X_{t+1} =\VEC x_{t+1}, \VEC Y^1_{t+1} = \VEC y^1_{t+1}, \VEC Y^2_{t+1} = \VEC y^2_{t+1}|\VEC y^1_{1:t},\VEC y^2_{1:t},\VEC u^1_{1:t},\VEC u^2_{1:t}) \notag \\
&= \mathds{P}(\VEC Y^1_{t+1}=\VEC y^1_{t+1}|\VEC X_{t+1} =\VEC x_{t+1})\mathds{P}(\VEC Y^2_{t+1}=\VEC y^2_{t+1}|\VEC X_{t+1} =\VEC x_{t+1}) \notag \\
&\times \mathds{P}^{g^1_{1:t},g^2_{1:t}}(\VEC X_{t+1} =\VEC x_{t+1}|\VEC y^1_{1:t},\VEC y^2_{1:t},\VEC u^1_{1:t},\VEC u^2_{1:t}) \notag \\
&=  \mathds{P}(\VEC Y^1_{t+1}=\VEC y^1_{t+1}|\VEC X_{t+1} =\VEC x_{t+1})\mathds{P}(\VEC Y^2_{t+1}=\VEC y^2_{t+1}|\VEC X_{t+1} =\VEC x_{t+1})  \notag \\
&\times \sum_{\VEC x_{t}}\Big[\mathds{P}(\VEC X_{t+1} =\VEC x_{t+1}|\VEC X_{t}=\VEC x_{t}, \VEC u^1_{t}, \VEC u^2_{t})\mathds{P}^{g^1_{1:t},g^2_{1:t}}(\VEC X_{t} =\VEC x_{t}|\VEC y^1_{1:t},\VEC y^2_{1:t},\VEC u^1_{1:t},\VEC u^2_{1:t})\Big], \label{eq:ex1.1}
\end{align}
where we used the dynamics and observation model to get the expression in \eqref{eq:ex1.1}. It can now be argued that in the last term in \eqref{eq:ex1.1}, we can remove the terms $\VEC u^1_t,\VEC u^2_t$ in the conditioning since they are functions of the rest of terms $\VEC y^1_{1:t},\VEC y^2_{1:t},\VEC u^1_{1:t-1},\VEC u^2_{1:t-1}$ in the conditioning. The last term in \eqref{eq:ex1.1} would then be 
\[  \mathds{P}^{g^1_{1:t},g^2_{1:t}}(\VEC X_{t} =\VEC x_{t}|\VEC y^1_{1:t},\VEC y^2_{1:t},\VEC u^1_{1:t-1},\VEC u^2_{1:t-1}),\] 
which is known to be independent of choice of control laws $g^1_{1:t},g^2_{1:t}$ \cite{KumarVaraiya}. Thus, $\pi_{t+1}$ is independent of the choice of control laws. For the sake of completeness, we provide a more detailed argument below.

The last term in \eqref{eq:ex1.1} can be written as
\begin{align}
&\mathds{P}^{g^1_{1:t},g^2_{1:t}}(\VEC X_{t} =\VEC x_{t}|\VEC y^1_{1:t},\VEC y^2_{1:t},\VEC u^1_{1:t},\VEC u^2_{1:t}) \notag \\
&= \mathds{P}^{g^1_{1:t},g^2_{1:t}}(\VEC X_{t} =\VEC x_{t}|\VEC y^1_{1:t},\VEC y^2_{1:t},\VEC u^1_{1:t-1},\VEC u^2_{1:t-1}) \notag \\ 
&= \frac{\mathds{P}^{g^1_{1:t},g^2_{1:t}}(\VEC X_{t} =\VEC x_{t}, \VEC Y^1_t=\VEC y^1_t, \VEC Y^2_t=\VEC y^2_t|\VEC y^1_{1:t-1},\VEC y^2_{1:t-1},\VEC u^1_{1:t-1},\VEC u^2_{1:t-1})}{\mathds{P}^{g^1_{1:t},g^2_{1:t}}(\VEC Y^1_t=\VEC y^1_t, \VEC Y^2_t=\VEC y^2_t|\VEC y^1_{1:t-1},\VEC y^2_{1:t-1},\VEC u^1_{1:t-1},\VEC u^2_{1:t-1})} \notag \\
&= \frac{\pi_t(\VEC x_t, \VEC y^1_t, \VEC y^2_t)}{\sum_{\VEC x'_t}\pi_t(\VEC x'_t, \VEC y^1_t, \VEC y^2_t)} \label{eq:ex1.1.2}
\end{align}

Combining \eqref{eq:ex1.1} and \eqref{eq:ex1.1.2} establishes that $\pi_{t+1}$ is a function only of $\pi_t$ and $\VEC z_{t+1} = (\VEC y^1_t, \VEC y^2_t, \VEC u^1_t, \VEC u^2_t)$. Further, the transformation form $(\pi_t, \VEC z_{t+1})$ to $\pi_{t+1}$ does not depend on the choice of control strategies.
\end{proof}
%%%%%%%%%%%%%%%%%%%%%%%%%%%%%%%%%%%%%%%%%%%%%%%%%%%%%%%%%%%%%%%%%%%%%%%%%%%%%%%%%%%%%%%%%%%%%%%%%%%%%%%%%%%%%%%%%
\section{Proof of Lemma \ref{lemma:asymmetric_delay}}\label{sec:asymmetric_delay}
\begin{proof}
 It is straightforward to verify that the structure of common and private information satisfies Assumption~\ref{assm:infoevolution}. We focus on the proof for Assumption~\ref{assm:separation}. For a realization $\VEC y^1_{1:t+1},\VEC y^2_{1:t}$, $\VEC u^1_{1:t},\VEC u^2_{1:t}$ of the common information at time $t+1$, the common information based belief  can be written as
\begin{align}
&\pi_{t+1}(\VEC x_{t+1},  \VEC y^2_{t+1}) = \mathds{P}^{g^1_{1:t},g^2_{1:t}}(\VEC X_{t+1} =\VEC x_{t+1},\VEC Y^2_{t+1} = \VEC y^2_{t+1}|\VEC y^1_{1:t+1},\VEC y^2_{1:t},\VEC u^1_{1:t},\VEC u^2_{1:t}) \notag \\
&= \mathds{P}(\VEC Y^2_{t+1} = \VEC y^2_{t+1}|\VEC X_{t+1} = \VEC x_{t+1})\mathds{P}^{g^1_{1:t},g^2_{1:t}}(\VEC X_{t+1} =\VEC x_{t+1} |\VEC y^1_{1:t+1},\VEC y^2_{1:t},\VEC u^1_{1:t},\VEC u^2_{1:t}) \notag \\
&= \mathds{P}(\VEC Y^2_{t+1} = \VEC y^2_{t+1}|\VEC X_{t+1} = \VEC x_{t+1})\frac{\mathds{P}^{g^1_{1:t},g^2_{1:t}}(\VEC X_{t+1} =\VEC x_{t+1}, \VEC Y^1_{t+1}=\VEC y^1_{t+1}|\VEC y^1_{1:t},\VEC y^2_{1:t},\VEC u^1_{1:t},\VEC u^2_{1:t})}{\sum_{\VEC x}\mathds{P}^{g^1_{1:t},g^2_{1:t}}(\VEC X_{t+1} =\VEC x, \VEC Y^1_{t+1}=\VEC y^1_{t+1}|\VEC y^1_{1:t},\VEC y^2_{1:t},\VEC u^1_{1:t},\VEC u^2_{1:t})} \label{eq:bayes1}
\end{align}
The numerator in the second term in \eqref{eq:bayes1} can be written as
\begin{align}
&\mathds{P}(\VEC Y^1_{t+1} = \VEC y^1_{t+1}|\VEC X_{t+1}=\VEC x_{t+1})\mathds{P}^{g^1_{1:t},g^2_{1:t}}(\VEC X_{t+1} =\VEC x_{t+1} |\VEC y^1_{1:t},\VEC y^2_{1:t},\VEC u^1_{1:t},\VEC u^2_{1:t}) \notag 
\end{align}
\begin{align}
&=\mathds{P}(\VEC Y^1_{t+1} = \VEC y^1_{t+1}|\VEC X_{t+1}=\VEC x_{t+1})\times \notag \\
&\sum_{\VEC x_t}\Big[\mathds{P}(\VEC X_{t+1}=x_{t+1}|\VEC X_t=\VEC x_t,\VEC u^1_t,\VEC u^2_t)\mathds{P}^{g^1_{1:t},g^2_{1:t}}(\VEC X_{t} =\VEC x_{t} |\VEC y^1_{1:t},\VEC y^2_{1:t},\VEC u^1_{1:t-1},\VEC u^2_{1:t-1})\Big] \notag
\end{align}
\begin{align}
&=\mathds{P}(\VEC Y^1_{t+1} = \VEC y^1_{t+1}|\VEC X_{t+1}=\VEC x_{t+1})\times \notag \\
&\sum_{\VEC x_t}\Big[\mathds{P}(\VEC X_{t+1}=\VEC x_{t+1}|\VEC X_t=\VEC x_t,\VEC u^1_t,\VEC u^2_t)\frac{\mathds{P}^{g^1_{1:t},g^2_{1:t}}(\VEC X_{t} =\VEC x_{t}, \VEC y^2_t |\VEC y^1_{1:t},\VEC y^2_{1:t-1},\VEC u^1_{1:t-1},\VEC u^2_{1:t-1})}{\mathds{P}^{g^1_{1:t},g^2_{1:t}}(\VEC y^2_t |\VEC y^1_{1:t},\VEC y^2_{1:t-1},\VEC u^1_{1:t-1},\VEC u^2_{1:t-1})}\Big] \notag \\
&=\mathds{P}(\VEC Y^1_{t+1} = \VEC y^1_{t+1}|\VEC X_{t+1}=\VEC x_{t+1})\sum_{\VEC x_t}\Big[\mathds{P}(\VEC X_{t+1}=\VEC x_{t+1}|\VEC X_t=\VEC x_t,\VEC u^1_t,\VEC u^2_t)\frac{\pi_t(\VEC x_t, \VEC y^2_t)}{\pi_t(\VEC y^2_t)}\Big] 
\end{align}
%&=\mathds{P}(\VEC Y^2_{t+1} = \VEC y^2_{t+1}|\VEC X_{t+1} = \VEC x_{t+1})\frac{\mathds{P}(\VEC Y^1_{t+1} = \VEC y^1_{t+1}|\VEC X_{t+1}=\VEC x_{t+1})\mathds{P}^{g^1_{1:t},g^2_{1:t}}(\VEC X_{t+1} =\VEC x_{t+1} |\VEC y^2_{1:t},\VEC u^1_{1:t},\VEC u^2_{1:t})}{\sum_{\VEC x}\mathds{P}(\VEC Y^1_{t+1} = \VEC y^1_{t+1}|\VEC X_{t+1}=\VEC x)\mathds{P}^{g^1_{1:t},g^2_{1:t}}(\VEC X_{t+1} =\VEC x |\VEC y^2_{1:t},\VEC u^1_{1:t},\VEC u^2_{1:t})} \notag 
%\end{align}
Similar expressions can be obtained for the denominator of the second term in \eqref{eq:bayes1} to get
\begin{align}
&\pi_{t+1}(\VEC x_{t+1},  \VEC y^2_{t+1}) = \mathds{P}(\VEC Y^2_{t+1} = \VEC y^2_{t+1}|\VEC X_{t+1} = \VEC x_{t+1})\times\notag\\
&\frac{\mathds{P}(\VEC Y^1_{t+1} = \VEC y^1_{t+1}|\VEC X_{t+1}=\VEC x_{t+1})\sum_{\VEC x_t}\Big[\mathds{P}(\VEC X_{t+1}=\VEC x_{t+1}|\VEC X_t=\VEC x_t,\VEC u^1_t,\VEC u^2_t)\pi_t(\VEC x_t, \VEC y^2_t)\Big]}{\mathds{P}(\VEC Y^1_{t+1} = \VEC y^1_{t+1}|\VEC X_{t+1}=\VEC x)\sum_{\VEC x'_t}\Big[\mathds{P}(\VEC X_{t+1}=\VEC x|\VEC X_t=\VEC x'_t,\VEC u^1_t,\VEC u^2_t)\pi_t(\VEC x'_t, \VEC y^2_t)\Big]}\notag \\
&=: F_t(\pi_t,\VEC y^1_{t+1},\VEC y^2_t, \VEC u^1_t,\VEC u^2_t) = F_t(\pi_t,\VEC z_{t+1})
\end{align}
%\begin{align}
%&=  \mathds{P}(\VEC Y^1_{t+1}=\VEC y^1_{t+1}|\VEC X_{t+1} =\VEC x_{t+1})\mathds{P}(\VEC Y^2_{t+1}=\VEC y^2_{t+1}|\VEC X_{t+1} =\VEC x_{t+1}) \times \notag \\
%&\sum_{\VEC x_{t}}\Big[\mathds{P}(\VEC X_{t+1} =\VEC x_{t+1}|\VEC X_{t}=\VEC x_{t}, \VEC u^1_{t}, \VEC u^2_{t})\mathds{P}^{g^1_{1:t},g^2_{1:t}}(\VEC X_{t} =\VEC x_{t}|\VEC y^1_{1:t},\VEC y^2_{1:t},\VEC u^1_{1:t},\VEC u^2_{1:t})\Big], \label{eq:ex1.1}
%\end{align}
\end{proof}
%%%%%%%%%%%%%%%%%%%%%%%%%%%%%%%%%%%%%%%%%%%%%%%%%%%%%%%%%%%%%%%%%%%%%%%%%%%%%55
\section{Proof of Lemma \ref{lemma:new_example}} \label{sec:new_example}
Assumption~\ref{assm:infoevolution} is clearly satisfied. We focus on Assumption~\ref{assm:separation}.
\emph{Case A:} For a realization $\VEC y^1_{1:t},\VEC y^2_{1:t-d}, \VEC u^1_{1:t-1}$ of the common information, the common information based belief in this case can be written as:
\begin{align}
&\pi_t(\VEC x_t, \VEC y^2_{t-d+1:t})= \mathds{P}^{g^1_{1:t-1}}(\VEC X_t =\VEC x_t, \VEC Y^2_{t-d+1:t}= \VEC y^2_{t-d+1:t}|\VEC y^1_{1:t},\VEC y^2_{1:t-d}, \VEC u^1_{1:t-1}) \notag \\
&= \sum_{\VEC x'_{t-d:t-1}}\Big[\mathds{P}(\VEC Y^2_{t-d+1:t}= \VEC y^2_{t-d+1:t}|\VEC X_{t-d+1:t-1}= \VEC x'_{t-d+1:t-1}, \VEC X_t = \VEC x_t)\notag \\
&\cdot \mathds{P}^{g^1_{1:t-1}}(\VEC X_t =\VEC x_t, \VEC X_{t-d:t-1}= \VEC x'_{t-d:t-1}|\VEC y^1_{1:t},\VEC y^2_{1:t-d}, \VEC u^1_{1:t-1})\Big] \label{eq:new_example1}
\end{align}
The first term in \eqref{eq:new_example1} depends only on the noise statistics. To see how  the second term in \eqref{eq:new_example1} is strategy independent, consider a centralized stochastic control problem with controller $1$ as the only controller where the state process is $\tilde{\VEC X}_t := (\VEC X_{t-d:t})$, the observation process is $\tilde{\VEC Y}_t := (\VEC Y^1_{t},\VEC Y^2_{t-d})$. The second term in \eqref{eq:new_example1} is simply the information state $\mathds{P}(\tilde{\VEC X}_t| \tilde{\VEC y}_{1:t}, \VEC u^1_{1:t-1})$ of this centralized stochastic control problem which is known to be strategy independent and satisfies an update equation of the form required by Lemma~\ref{lemma:new_example} \cite{KumarVaraiya}.

\emph{Case B:} Using arguments similar to those in Case A, the common information based belief $\pi_t$ for a realization $\VEC y^1_{1:t-1},\VEC y^2_{1:t-d}, \VEC u^1_{1:t-1}$ of the common information can be written as:
\begin{align}
&\pi_t(\VEC x_t, \VEC y^1_t, \VEC y^2_{t-d+1:t}) = \sum_{\VEC x'_{t-d:t-1}}\Big[\mathds{P}(\VEC Y^1_t= \VEC y^1_t,\VEC Y^2_{t-d+1:t}= \VEC y^2_{t-d+1:t}|\VEC X_{t-d+1:t-1}= \VEC x'_{t-d+1:t-1}, \VEC X_t = \VEC x_t)\notag \\
&\cdot \mathds{P}^{g^1_{1:t-1}}(\VEC X_t =\VEC x_t, \VEC X_{t-d:t-1}= \VEC x'_{t-d:t-1}|\VEC y^1_{1:t-1},\VEC y^2_{1:t-d}, \VEC u^1_{1:t-1})\Big] \label{eq:new_example2}
\end{align}
The second term in \eqref{eq:new_example2} is
\begin{align}
\frac{\mathds{P}(\VEC y^2_{t-d}|\VEC x_{t-d})\mathds{P}^{g^1_{1:t-1}}(\VEC X_t =\VEC x_t, \VEC X_{t-d:t-1}= \VEC x'_{t-d:t-1}|\VEC y^1_{1:t-1},\VEC y^2_{1:t-d-1}, \VEC u^1_{1:t-1})}{\mathds{P}^{g^1_{1:t-1}}(\VEC Y^2_{t-d} =\VEC y^2_{t-d}|\VEC y^1_{1:t-1},\VEC y^2_{1:t-d-1}, \VEC u^1_{1:t-1})}
\end{align}
Both the numerator and the denominator can be shown to be strategy independent using the transformation to centralized stochastic control problem described in case A.

%%%%%%%%%%%%%%%%%%%%%%%%%%%%%%%%%%%%%%%%%%%%%%%%%%%%%%%%%%%%%%%%%%%%%%%%%%%%%
\section{Proof of Lemma \ref{lemma:noisy_global}}\label{sec:noisy_global}

For a realization $y^0_{1:t+1},\VEC u^1_{1:t}, \VEC u^2_{1:t}$ of the common information at time $t+1$, the belief $\pi_{t+1}$ is given as
\begin{align}
&\pi_{t+1}(x^0,x^1,x^2) = \mathds{P}^{g^1_{1:t-1},g^2_{1:t-1}}(X^0_{t+1}=x^0,X^1_{t+1}=x^1,X^2_{t+1}=x^2|y^0_{1:t+1},\VEC u^1_{1:t}, \VEC u^2_{1:t})
\end{align}
\begin{align}
&=\frac{\mathds{P}^{g^1_{1:t-1},g^2_{1:t-1}}(X^0_{t+1}=x^0,X^1_{t+1}=x^1,X^2_{t+1}=x^2, Y^0_{t+1}= y^0_{t+1}|y^0_{1:t},\VEC u^1_{1:t}, \VEC u^2_{1:t})}{\mathds{P}^{g^1_{1:t-1},g^2_{1:t-1}}(Y^0_{t+1}= y^0_{t+1}|y^0_{1:t},\VEC u^1_{1:t}, \VEC u^2_{1:t})}\notag\\
&=\frac{\mathds{P}(Y^0_{t+1}= y^0_{t+1}|X^0_{t+1}= x^0_{t+1})\mathds{P}^{g^1_{1:t-1},g^2_{1:t-1}}(X^0_{t+1}=x^0,X^1_{t+1}=x^1,X^2_{t+1}=x^2|y^0_{1:t},\VEC u^1_{1:t}, \VEC u^2_{1:t})}{ \sum_{x}\mathds{P}(Y^0_{t+1}= y^0_{t+1}|X^0_{t+1}=x)\mathds{P}^{g^1_{1:t-1},g^2_{1:t-1}}(X^0_{t+1}=x|y^0_{1:t},\VEC u^1_{1:t}, \VEC u^2_{1:t})}\label{eq:noisyglobal1}
\end{align}
The control strategy dependent term in the numerator in \eqref{eq:noisyglobal1} can be written as
\begin{align}
&\mathds{P}^{g^1_{1:t-1},g^2_{1:t-1}}(X^0_{t+1}=x^0,X^1_{t+1}=x^1,X^2_{t+1}=x^2|y^0_{1:t},\VEC u^1_{1:t}, \VEC u^2_{1:t}) \notag \\
&= \sum_{x'}\Big[\mathds{P}(X^0_{t+1}=x^0,X^1_{t+1}=x^1,X^2_{t+1}=x^2|X^0_t=x', \VEC u^1_t,\VEC u^2_t)\notag \\&\cdot\mathds{P}^{g^1_{1:t-1},g^2_{1:t-1}}(X^0_{t}=x'|y^0_{1:t},\VEC u^1_{1:t-1}, \VEC u^2_{1:t-1})\Big]\notag \\
&= \sum_{x'}\mathds{P}(X^0_{t+1}=x^0,X^1_{t+1}=x^1,X^2_{t+1}=x^2|X^0_t=x', \VEC u^1_t,\VEC u^2_t)\pi_t(x')\label{eq:noisyglobal2}
\end{align}
Similarly, the control strategy dependent term in the denominator in \eqref{eq:noisyglobal1} can be written as
\begin{align}
\mathds{P}^{g^1_{1:t-1},g^2_{1:t-1}}(X^0_{t+1}=x|y^0_{1:t},\VEC u^1_{1:t}, \VEC u^2_{1:t}) =\sum_{x''}\mathds{P}(X^0_{t+1}=x|X^0_t=x'', \VEC u^1_t,\VEC u^2_t)\pi_t(x'')\label{eq:noisyglobal3}
\end{align}
Substituting \eqref{eq:noisyglobal2} and \eqref{eq:noisyglobal3} in \eqref{eq:noisyglobal1} establishes the lemma.
%%%%%%%%%%%%%%%%%%%%%%%%%%%%%%%%%%%%%%%%%%%%%%%%%%%%%%%%%%%%%%%%%%%%%%%%%%%%%%%%%%%%%%%%%%%%%%
\section{Proof of Lemma \ref{lemma:uncontrolled}}\label{sec:uncontrolled_app}

 Consider a realization
$\VEC c_{t}$ of the common information $\VEC C_{t}$ at time $t$. 
Given the realization of the common information based belief $\pi_t$, we can find the joint conditional distribution on $(\VEC X_t, \VEC P^1_t, \VEC P^2_t, \VEC X_{t+1}, \VEC P^1_{t+1}, \VEC P^2_{t+1}, \VEC Z_{t+1})$ conditioned on the common information at time $t$ as follows:
\begin{align}
&\mathds{P}(\VEC x_t, \VEC p^1_t, \VEC p^2_t, \VEC x_{t+1}, \VEC p^1_{t+1}, \VEC p^2_{t+1}, \VEC z_{t+1}|\VEC c_t) \notag\\
&= \sum_{\VEC y^1_{t+1},\VEC y^2_{t+1}}\mathds{P}(\VEC x_t, \VEC p^1_t, \VEC p^2_t, \VEC x_{t+1}, \VEC p^1_{t+1}, \VEC p^2_{t+1}, \VEC z_{t+1}, \VEC y^1_{t+1},\VEC y^2_{t+1}|\VEC c_t) \notag \\
&=\sum_{\VEC y^1_{t+1},\VEC y^2_{t+1}}\Big[\mathds{1}_{\{\zeta_{t+1}(\VEC p^1_t, \VEC p^2_t, \VEC y^1_{t+1}, \VEC y^2_{t+1})=\VEC z_{t+1}\}} \mathds{1}_{\{\xi^1_{t+1}(\VEC p^1_t,  \VEC y^1_{t+1}) = \VEC p^1_{t+1}\}}\mathds{1}_{\{\xi^2_{t+1}(\VEC p^2_t,   \VEC y^2_{t+1}) = \VEC p^2_{t+1}\}} \notag \\
&\times\mathds{P}(\VEC y^1_{t+1},\VEC y^2_{t+1}|\VEC x_{t+1})\mathds{P}(\VEC x_{t+1}|\VEC x_t)\pi_t(\VEC x_t, \VEC p^1_t, \VEC p^2_t)\Big]\label{eq:unc_appendixeq1}
\end{align}
Note that in addition to the arguments on the left side of conditioning in \eqref{eq:unc_appendixeq1}, we only need $\pi_t$  to evaluate the right hand side of \eqref{eq:unc_appendixeq1}. 
\par
We can now consider the common information based belief at time $t+1$,
\begin{align}
&\pi_{t+1}(\VEC x_{t+1}, \VEC p^1_{t+1}, \VEC p^2_{t+1}) = \mathds{P}(\VEC x_{t+1}, \VEC p^1_{t+1}, \VEC p^2_{t+1}|\VEC c_{t+1}) \notag \\
&= \mathds{P}(\VEC x_{t+1}, \VEC p^1_{t+1}, \VEC p^2_{t+1}|\VEC c_{t},\VEC z_{t+1}) \notag \\
&= \frac{\mathds{P}(\VEC x_{t+1}, \VEC p^1_{t+1}, \VEC p^2_{t+1},\VEC z_{t+1}|\VEC c_{t})}{\mathds{P}(\VEC z_{t+1}|\VEC c_{t})}\label{eq:unc_appendixeq2}
\end{align}
 The numerator and denominator of \eqref{eq:unc_appendixeq2} are both marginals of the probability in \eqref{eq:unc_appendixeq1}. Using \eqref{eq:unc_appendixeq1} in \eqref{eq:unc_appendixeq2}, gives $\pi_{t+1}$ as a function of $\pi_t,\VEC z_{t+1}$.

%%%%%%%%%%%%%%%%%%%%%%%%%%%%%%%%%%%%%%%%%%%%%%%%%%%%%%%%%%%%%%%%%%%%%%%%%%%%%%%%%%%%%%%55%%%%%
\end{document}